\documentclass[journal]{IEEEtran}
\usepackage{amssymb}
\usepackage{amsmath}
\usepackage{breqn}
 \usepackage{amsthm}
\usepackage{array}
\usepackage{mathtools}
\usepackage{graphicx}
\usepackage{cite}
\usepackage{gensymb}
\usepackage{centernot}
\usepackage{enumerate}
\usepackage{pifont}
\usepackage{mleftright}
\usepackage{thmtools,thm-restate}
\usepackage[parfill]{parskip}
\usepackage{verbatim}
\usepackage[hidelinks]{hyperref}
\usepackage[flushleft]{threeparttable}
\usepackage{subcaption}

 \newcommand\x{\mathbf{x}}

 \newcommand\I{J}
  \newcommand\J{\mathbf{J}}

  \renewcommand\d{\mathbf{d}}
  
  \newcommand\s{\mathbf{s}}
    \newcommand\cc{\mathbf{c}}
        \newcommand\C{\mathcal{C}}

        \newcommand\OO{\mathcal{O}}

        \newcommand\A{\mathcal{A}}
        \newcommand\B{\mathcal{B}}
  
 \newcommand\bx{\boldsymbol{\xi}}

  \newcommand\bX{\boldsymbol{\Xi}}

  \newcommand\bI{\mathcal{J}}
   \newcommand\bd{\boldsymbol{\delta}}
    
  \newcommand\var{\mathrm{Var}}
 
 \newcommand\uu{\mathbf{u}}
 \newcommand\R{\mathbb{R}}
  \newcommand\N{\mathbb{N}}
  \newcommand\Ex{\mathbb{E}}
  
  \newcommand\K{\mathcal{K}}
  \newcommand{\notimplies}{%
  	\mathrel{
  		\ooalign{\hidewidth$\not{\phantom{=}}$\hidewidth\cr$\implies$}
  	}
  }

  \newcommand\be{\begin{enumerate}}
  \renewcommand\ne{\end{enumerate}}

  \DeclareMathOperator*{\argmax}{arg\,max}
   \DeclareMathOperator*{\argmin}{arg\,min}
 
\newtheorem{theorem}{Theorem}[section]
\newtheorem{lemma}[theorem]{Lemma}
\newtheorem{corollary}[theorem]{ Corollary}
\newtheorem{proposition}[theorem]{Proposition}
\theoremstyle{definition}
\newtheorem{definition}[theorem]{Definition}
\theoremstyle{remark}
\newtheorem{remark}{Remark}[section]
\begin{document}

\title{Source Localisation Using Binary Measurements}
	
\author{ 
\IEEEauthorblockN{Daniel D. Selvaratnam\IEEEauthorrefmark{1}\textsuperscript{1}, Iman Shames\IEEEauthorrefmark{1}, Jonathan H. Manton\IEEEauthorrefmark{1}, \IEEEmembership{Fellow,~IEEE}, and Branko Ristic\IEEEauthorrefmark{2}} \\
\IEEEauthorblockA{\IEEEauthorrefmark{1} Dept. Electrical and Electronic Engineering, University of Melbourne} \\
\IEEEauthorblockA{\IEEEauthorrefmark{2} School of Engineering, RMIT University}
\thanks{This research is supported by DST Group under Collaborative Research Agreements MYIP \#6874, MYIP \#5923, and by the Defence Science Institute as an initiative of the State Government of Victoria.}
\thanks{\textsuperscript{1} Corresponding author: \texttt{dselvaratnam@student.unimelb.edu.au}.}
\thanks{Preprint submitted to \emph{IEEE Transactions on Signal Processing.}}
}


\maketitle

\begin{abstract}
	
This paper considers the problem of localising a stationary signal source using a team of mobile agents which only take binary measurements.  Background false detection rates and missed detection probabilities are incorporated into the framework. A Bayesian estimation algorithm that discretises the search environment is employed, and analytical convergence and consistency results for this are derived. Fisher Information is then used as a metric for the design of optimal agent geometries. Knowledge of the probability of detection as a function of the source and agent locations is assumed in the analysis, with special attention given to range-dependent functions. The behaviour of the algorithm under inexact knowledge of the probability of detection is also analysed. Finally, simulation results are presented to demonstrate the effectiveness of the algorithm.
\end{abstract}

\section{Introduction}
\label{sec:intro}
Source localisation involves estimating the location of a signal source using measurements from a set of available sensors. Regardless of its type, a single bit constitutes the minimum amount of information that can be extracted from a signal. For example, when dealing with chemical or radiological sources, this may correspond to detecting the presence or absence of particles of interest \cite{ristic2016,vergassola2007}. 
In other situations, sensors may be required to process raw measurement data locally and report a binary outcome to a fusion centre \cite{chamberland2003, chen2006}. In general, any continuous or discrete-valued signal can be converted into a binary one via the use of a threshold. 
This is often desirable for applications with limited resources, because binary data demands less memory, communication bandwidth, and energy from the agents involved \cite{niu2004}. In this paper, we specifically consider the localisation of a stationary source using binary measurements obtained from a team of mobile agents. 

Consistent with standard practice, the measurements are treated as random variables taking values in $\{0,1\}$, to capture the effects of sensor noise and environmental uncertainty. Given suitable models for the signal propagation and sensors, the probability of obtaining a detection (i.e. of measuring a 1) becomes a well-defined function of the source and agent locations. See \cite{vijayakumaran2007, ristic2016, ristic2015a,shoari2014,niu2004} and Section \ref{sec:sim} of this paper for specific examples of such functions constructed for different types of signals and sensors. A novel aspect of our analysis is that it assumes an arbitrary probability-of-detection function, subject to mild conditions. Thus, the algorithms and results in the paper apply to a large class of measurement models and localisation scenarios. The probability-of-detection function is initially assumed to be fully known, but this is later relaxed by analysing the performance of the algorithm when only an envelope for this function is known. Background false detection rates and missed detection probabilities are incorporated naturally into our framework. 

Both Bayesian and classical parameter estimation techniques have been applied to solve the source localisation problem. We adopt the former, which has the advantage of incorporating prior knowledge about the source location, and of maintaining an entire posterior probability distribution rather than just a single estimate. Furthermore, a Bayesian framework permits the recursive addition of new measurements to update the posterior, without reprocessing past measurements. A disadvantage of this approach is that every iteration requires the computation of integrals that, in general, have no analytic solution. We obtain a tractable approximation by discretising the exploration region, thereby replacing the integrals with sums and generating a discrete posterior instead of a continuous one. This technique is well-known, however we believe the accompanying analysis to be novel. In particular, we explicitly consider the effects of finite discretisation by identifying points at which the discrete posterior is guaranteed to vanish asymptotically, and establish a relationship between the decay and Kullback-Leibler (KL) divergence. The analysis leads directly to conditions on measurement locations that guarantee sufficient information is being extracted by the agents. We then extend this by choosing measurement locations to maximise the determinant of the Fisher Information Matrix (FIM) \cite[Section 4.3.3.1]{vantrees2013}. This is a widely adopted performance criterion known as \emph{D-optimality} \cite{chaloner1995}. Focusing on the case where the probability of detection depends solely on distance, the resulting D-optimal geometries mirror the results of \cite{bishop2010} for range-only sensors. We then show how the knowledge of these geometries can be exploited via a control strategy by guiding the agents into formation about an estimated source location. 

Importance sampling is an alternative Bayesian technique which uses random sampling to numerically evaluate the required integrals \cite[Chapter 14]{robert2004}, \cite{ristic2004}. The posterior is approximated by a weighted set of samples or \emph{particles}, and well known results prove the convergence of this approximation to the true posterior as the number of particles approaches infinity. In practice however, only a finite number of particles can ever be used. In Section \ref{sec:equivalence}, we establish that importance sampling with a finite number of particles can be treated as a special case of discretising the exploration region, under the appropriate choice of discretisation points and prior. Our approach is more general because it allows the discretisation points to be chosen arbitrarily. Another key difference is that we analyse convergence over time, using only a finite number of particles. Particle filtering~\cite{ristic2004} extends important sampling to estimate a time-varying state based on an assumed dynamic model. Particle filtering does not fit within the framework of this paper, because it requires the particles to be propagated according to the dynamic model, and re-sampled every time-step. However, when only a stationary source is involved, nothing is gained by performing these additional steps.  Standard importance sampling/discretisation therefore remains a more appropriate choice for the problem at hand, and the simulation results in Section \ref{sec:sim} illustrate this.

A preliminary version of the results presented in this paper appeared in the conference proceedings \cite{selvaratnam2017b}, which developed limited posterior convergence results focusing on two special cases: measurements taken at a single location with an arbitrarily located source, and measurements taken at a periodic sequence of locations assuming a source coincident with one of the chosen discretisation points. In this journal paper, we extend the latter to include an arbitrarily located source, and strengthen all the results to almost-sure convergence. Other new theoretical developments include the relationship with KL divergence, D-optimal location optimisation, and the analysis relating to inexact knowledge of the probability of detection.
We now present a brief review of other relevant works in the literature, dividing them into Bayesian and classical approaches. 

A Bayesian approach is adopted in \cite{vergassola2007} to localise a chemical source using a single mobile agent which detects the presence or absence of an odour. As in our own work, the search region is discretised to approximate the posterior, however a theoretical convergence analysis is not offered. Rather, the focus of the paper is on a search strategy based on maximising the rate of entropy reduction.
Importance sampling is employed in \cite{ristic2016} for source localisation with binary measurements, using a propagation model based on turbulent dispersion in the atmosphere. Their approach accommodates an unknown particle release rate by using Rao-Blackwellisation \cite{ristic2017} to estimate it explicitly. The same Bayesian algorithm underpins \cite{vergassola2007, ristic2016, ristic2017}, and our own work. We emphasize that our contribution is not to propose a new estimation algorithm, but rather to provide a rigorous treatment of the inevitable effects of discretisation, supplemented with numerical results. A search for multiple stationary targets is considered in \cite{hu2013a}, which considers a discrete environment to begin with, and assumes the agents directly observe the occupancy state of each cell with given false and missed detection probabilities. Since binary measurements are typically generated by means of a threshold, several works address the problem of designing threshold levels. These include \cite{ristic2015a}, which studies the best achievable localisation accuracy using a binary sensor network, under a Gaussian plume propagation model. Threshold levels and sensor placement are investigated using the Bayesian Information Matrix (BIM), and the resulting theoretical error bounds are compared with the performance of the Metropolis-Hastings estimation algorithm. The tracking of a moving source using binary measurements is considered in \cite{vemula2007}, which uses particle filtering to estimate the source location, and proposes a heuristic for adaptively designing sensor threshold levels. This is extended to multi-bit measurements in \cite{ozdemir2008}, which focuses on adaptively designing quantisation thresholds based on the Bayesian Information Matrix. 

Classical approaches treat the source location as a deterministic but unknown parameter, rather than a random variable. They tend to focus on constructing estimators rather than maintaining a probability distribution. A maximum likelihood estimator is proposed in \cite{vijayakumaran2007} for localising a diffusive source using binary measurements. That algorithm seeks to estimate a two-dimensional source location, time of signal emission, and several other model parameters via Fisher Scoring, a modified Newton method for maximising the likelihood function. Convergence guarantees are obtained as the number of sensors goes to infinity. Since each iteration requires reprocessing the entire batch of measurements, \cite{vijayakumaran2007} also proposes a real-time approximate algorithm to avoid this.  We compare the complexity and numerical performance of these maximum likelihood approaches with our own in Section \ref{sec:sim}. A set of different estimators are constructed in \cite{shoari2010} without the use of any probability of detection model, but assuming noise free detections. Such model independent approaches clearly require less prior information, but typically display worse performance \cite{shoari2014}. As in the Bayesian case, the design of binary quantisation thresholds based on the FIM is studied in \cite{shoari2014}. Thresholds for multi-bit quantisation are studied in \cite{niu2006}, which also compares the resulting theoretical error bounds with the performance of the maximum likelihood estimator and a second estimator that takes a weighted average of the sensor locations. 

The remainder of this paper is structured as follows. The problem is formulated mathematically in Section \ref{sec:prob}. The estimation algorithm is then developed and analysed in Section \ref{sec:theorems}. Section \ref{sec:fish} derives D-optimal measurement locations, and Section \ref{sec:lambda} considers the implications of having inexact knowledge of the probability-of-detection function. A numerical example and simulation results are presented in \ref{sec:sim}. Closing remarks are then made in Section \ref{sec:conc}.
\section{Problem formulation}
\label{sec:prob}
For the remainder of this paper, we adopt the convention $ \N = \{1,2,...\}$ and define $\N_k := \{1,...,k\}$. We also use the notation $(a_k)_{k \in \N} \subset A$ to denote $a_k \in A$ for all $k \in \N$. 

Consider a team of $N$ agents exploring $\R^{q}$, where \hbox{$q \in \{2,3\}$.} 
Let agent $i$ have position $\x_i(t) \in \R^{q}$, which evolves in continuous time. We assume that all agents know their own position with respect to the same co-ordinate frame, and are equipped with identical sensors. The agents must search a compact region $S \subset \R^{q}$ for a source located at $\s \in S$.
Together, the team of agents take a sequence of measurements $(d_k)_{k \in \N} \subset \{0,1\}$ at a corresponding sequence of locations $(\bx_k)_{k \in \N} \subset \R^q$. The measurement pairs $(\bx_k,d_k)$ are transmitted in real-time to a fusion centre, where they are processed on arrival. The subscripts $k$ index the measurements according to the order in which are processed by the fusion centre. Note that the fusion centre is agnostic to the identity of the observing agent. Thus $\bx_k \in \{ \x_i(t_k) \mid i \in \N_N\}$, where $t_k \geq 0$ denotes the time at which reading $d_k$ was taken.  

We model $d_1,d_2,...$ as random variables that are conditionally independent of each other, given the source location.
We assume the probability of receiving a detection is a known continuous function $\ell: \R^{q} \times \R^{q} \to (0,1)$ of the source and agent locations. Initially, we make no further assumptions about $\ell$. Let $\ell(\R^q,\R^q) \subset (0,1)$ denote its image. Observe that there is always some non-zero probability of failing to detect the signal, as well as a non-zero background false detection probability regardless of where the source is.
 Having defined $\ell$, the probability of obtaining the reading $d_k$ from an agent at position $\bx_k$ when the source location is $\s$, is given by the likelihood function
\begin{equation} g(d_k \mid \s; \bx_k) = \ell(\s,\bx_k)^{d_k} \left[ 1 - \ell(\s,\bx_k) \right]^{1 - d_k}. \label{eq:gfun} \end{equation}

\section{Estimation of source location}
\label{sec:theorems}
The estimation algorithm is developed in this section. 
We treat $\s$ as a random variable, drawn from some prior distribution $p_0$ over $S$. Bayesian techniques allow us to compute the posterior probability density of $\s$, given the history of measurements $d_{1:k} := (d_1,...,d_k)$ and corresponding agent poses $\bx_{1:k} = (\bx_1,...,\bx_k)$. Bayes rule gives us a recursive description of this posterior density
\begin{dmath}p_k(\s \mid d_{1:k} ; \bx_{1:k}) = \dfrac{g(d_k \mid \s; \bx_k) p_{k-1}(\s \mid d_{1:k-1} ; \bx_{1:k-1})}{ \displaystyle \int_S g(d_k \mid \s'; \bx_k) p_{k-1}(\s'\mid d_{1:k-1} ; \bx_{1:k-1})d \s'}, \label{eq:bayesmapcont} \end{dmath}
where the recursion is initialized with $p_0(\s)$. 

Although \eqref{eq:bayesmapcont} is exact, the integrals involved do not, in general, have a closed-form, analytic solution. In order to work with arbitrary $\ell$, the posterior must be approximated, and \eqref{eq:bayesmapcont} computed numerically. 
To tackle this, we discretise $S$ into a finite set of distinct points \hbox{$ \C := \{\cc_1,...,\cc_M\}$}, the elements of which we refer to as \emph{centres}. If it is known that $\s \in \C$, then this yields a discrete version of the Bayes recursion \eqref{eq:bayesmapcont},
\begin{dmath}\hat{p}_k(i \mid d_{1:k} ; \bx_{1:k}) = \dfrac{g(d_k \mid \cc_i; \bx_k) \hat{p}_{k-1}(i \mid d_{1:k-1} ; \bx_{1:k-1})}{\sum\limits_{j=1}^M g(d_k \mid \cc_j; \bx_k) \hat{p}_{k-1}(j \mid d_{1:k-1} ; \bx_{1:k-1})}, \label{eq:bayesmap} \end{dmath}
which is initialized with a discrete prior $\hat{p}_0(i)$. Without loss of generality, we assume $\hat{p}_0:\N_M \to (0,1)$, noting that any $\cc_i$ for which $\hat{p}_0(i) = 0$ can simply be omitted. For the more general case where $\s \in S$ is arbitrary, given a particular choice of centres, we can define a set of cells $C_1,...,C_M$, such that \be
\item each $C_i \subset S$ is connected, and $\cc_i \in C_i$ for all $i$
\item $S = \bigcup_{i=1}^M C_i$
\item $C_i$ and $C_j$ are interior disjoint for all $i \neq j$. \ne This lends the following interpretation to the discrete posterior: $$ \hat{p}_k(i \mid d_{1:k} ; \bx_{1:k}) \approx \mathrm{Pr}(\s \in C_i \mid d_{1:k}).$$
As an example, the centres and cells can be chosen to form a grid or, more generally, a Voronoi diagram. Alternatively, if the centres are sampled from a probability distribution, we show below that this corresponds to importance sampling with a finite number of particles under the appropriate choice of $\hat{p}_0$.
\subsection{Relationship to importance sampling}
\label{sec:equivalence}
Consider the posterior mean
$$ \hat{\s}_k = \int_S \s p_k(\s \mid d_{1:k} ; \bx_{1:k}) d \s .$$
This can be approximated numerically via
$$ \hat{\s}_k \approx  \frac{1}{M} \sum_{i=1}^M \s_i, $$
where each $\s_i \sim p_k(\s \mid d_{1:k} ; \bx_{1:k}) $. However, we do not have a closed form expression for $p_k$, so we are unable to sample from it directly. Importance sampling assumes the ability to sample from some other more convenient density $\phi:S \to [0,\infty)$, known as the \emph{importance density}. The importance density can be arbitrary, but its support must contain the support of $p_k$.  
The expectation can then be computed according to
$$ \hat{\s}_k \approx \sum_{i=1}^M \hat{w}^i_k \cc_i,$$ where
$ \cc_i \sim \phi(\s) $, and
$$ w^i_k = \frac{p_k(\cc_i \mid d_{1:k} ; \bx_{1:k})}{\phi(\cc_i)}, \quad \hat{w}^i_k = \frac{w_i}{\sum_{j=1}^M w_j}.$$
Here, the $\cc_i$ are referred to as \emph{particles}, and the $\hat{w}_i$ as \emph{weights}.  Note that $\hat{w}^i_k$ is the normalized version of $w^i_k$. 
Recalling \eqref{eq:bayesmapcont}, and defining $$\nu_k := \int_S g(d_k \mid \s'; \bx_k) p_{k-1}(\s'\mid d_{1:k-1} ; \bx_{1:k-1})d \s',$$ we see that
$$p_k(\cc_i \mid d_{1:k} ; \bx_{1:k}) = \dfrac{g(d_k \mid \cc_i; \bx_k) p_{k-1}(\cc_i \mid d_{1:k-1} ; \bx_{1:k-1})}{ \nu_k}.$$
The un-normalized weights therefore obey the recursive relationship
$$ w^i_k = \frac{g(d_k \mid \cc_i; \bx_k) p_{k-1}(\cc_i \mid d_{1:k-1} ; \bx_{1:k-1})}{ \nu_k \phi(\cc_i) } = \frac{g(d_k \mid \cc_i; \bx_k)}{ \nu_k } w^i_{k-1}.$$
Applying the normalization,
\begin{align*}
\hat{w}^i_k & = \dfrac{ g(d_k \mid \cc_i; \bx_k) w^i_{k-1} }{\nu_k \displaystyle \sum_{j=1}^M \dfrac{g(d_k \mid \cc_j; \bx_k) w^j_{k-1}}{ \nu_k }}  = \dfrac{ g(d_k \mid \cc_i; \bx_k) w^i_{k-1} }{\displaystyle \sum_{j=1}^M g(d_k \mid \cc_j; \bx_k) w^j_{k-1}}.
\end{align*}
Finally, letting $W_k = \sum_{j=1}^M w^j_k,$ we see that
\begin{align*}
\hat{w}^i_k & = \dfrac{\displaystyle g(d_k \mid \cc_i; \bx_k)  \frac{w^i_{k-1}}{W_{k-1}} }{\displaystyle \sum_{j=1}^M g(d_k \mid \cc_j; \bx_k) \frac{w^j_{k-1}}{W_{k-1}}} = \dfrac{ g(d_k \mid \cc_i; \bx_k) \hat{w}^i_{k-1} }{\displaystyle \sum_{j=1}^M g(d_k \mid \cc_j; \bx_k) \hat{w}^j_{k-1}}.
\end{align*}
Noting that this recursion is identical to \eqref{eq:bayesmap}, we see that the weight $\hat{w}^i_k$ obeys the same update rule as $\hat{p}_k(i \mid d_{1:k};\bx_{1:k})$. The initial weights are given by
\begin{equation} \hat{w}^i_0 = \frac{p_0(\cc_i)}{\displaystyle \phi(\cc_i)\sum_{j=1}^M \frac{p_0(\cc_j)}{\phi(\cc_j)}}. \label{eq:w0} \end{equation}
In Section \ref{sec:theorems}, no assumption is made about how the centres are chosen, and an arbitrary discrete prior is assumed. 
If we choose the $\cc_i \in \mathcal{C}$ by sampling from importance density, and initialize $\hat{p}_0 = \hat{w}^i_0$ above, then the discretised approach of Section \ref{sec:theorems} is identical to importance sampling. Thus, importance sampling with a finite number of particles becomes a special case of discretisation. 
This is stated formally below. 
\begin{theorem}[Importance Sampling] \label{th:equivalence}
	Let $\cc_i \sim \phi(\s)$ for all $i \in \N_M$, where $\phi:S \to [0,\infty)$ is a probability distribution that satisfies
	$$\phi(\s) = 0 \implies  \forall k \geq 0,\ p_k(\s \mid d_{1:k} ; \bx_{1:k}) = 0.$$ Furthermore, for all $i \in \N_M$, let
	$$\hat{p}_0(i) =  \frac{p_0(\cc_i)}{\displaystyle \phi(\cc_i)\sum_{j=1}^M \frac{p_0(\cc_j)}{\phi(\cc_j)}}. $$
	Then $\hat{p}_k(i \mid d_{1:k};\bx_{1:k}) = \hat{w}^i_k$ for all $i \in \N_M$ and all $k \geq 0$. 
\end{theorem}
\subsection{Definitions of fundamental quantities}
Informally, the requirement of posterior consistency means that the posterior should become increasingly concentrated about the source location as $k \to \infty$. A precise definition and discussion of consistency can be found in \cite[Section 4.1.1]{ghosh2006}. We adopt the following definition, specific to this problem.
\begin{definition}[Posterior Consistency] \label{def:consistency}
The posterior $\hat{p}_k$ is consistent if 
\begin{equation} \s \notin C_i \implies \lim_{k \to \infty} \hat{p}_k(i \mid d_{1:k} ; \bx_{1:k}) = 0 \text{ a.s.} \label{eq:def}, \end{equation}
where $C_i$ is the cell corresponding to centre $\cc_i$. 
\end{definition}
If the source does not lie on the boundary between two cells, \eqref{eq:def} is equivalent to 
$$ \s \in C_j \implies \lim_{k \to \infty} \hat{p}_k(j \mid d_{1:k} ; \bx_{1:k}) = 1 \text{ a.s.} .$$
The set
$$ \mathcal{O} := \left\{ i \in \N_M \mid \lim_{k \to \infty} \hat{p}_k(i \mid d_{1:k} ; \bx_{1:k}) = 0 \text{ a.s. } \right\}$$ is of clear interest, because consistency is also equivalent to ${\mathcal{O} = \{ i \in \N_M \mid \s \notin C_i \}}$. Note that for any pair of cells $(i,j)$, successive iterations of $\eqref{eq:bayesmap}$ up to time-step $n$ gives us 
\begin{equation} \dfrac{\hat{p}_n(i \mid d_{1:n} ; \bx_{1:n})}{\hat{p}_n(j \mid d_{1:n} ; \bx_{1:n})} = \dfrac{\hat{p}_0(i)}{\hat{p}_0(j)} \prod_{k=1}^n \frac{g(d_k \mid \cc_i; \bx_k)}{g(d_k \mid \cc_j; \bx_k)}. \label{eq:pratio} \end{equation}
We accordingly define the likelihood ratio \begin{equation} Z^{(i,j)}_k :=  \frac{ g(d_k \mid \cc_i; \bx_k)}{ g(d_k \mid \cc_j ; \bx_k)}, \label{eq:Xi} \end{equation}
which is the ratio of the probability of obtaining a reading $d_k$ with the source at $\cc_i$ to the probability of obtaining it with the source at $\cc_j$. The expected value of the log-likelihood ratio conditioned on the source location 
\begin{align} \mu^{(i,j)} _ k  & :=  \Ex\left[ \ln Z^{(i,j)}_k \mid \s \right] \\ & = \sum_{d=0}^1 \ln \left[\frac{ g(d \mid \cc_i; \bx_k)}{ g(d \mid \cc_j ; \bx_k)}\right]g(d \mid \s; \bx_k), \label{eq:rawmuk} \end{align} 
will play a key role in the subsequent analysis. Note that the value $\mu^{(i,j)}_k$ depends on $\cc_i, \cc_j, \s$ and $\bx_k$. 
Following from \eqref{eq:rawmuk}, this relationship can be written as
\begin{equation} \mu^{(i,j)} _ k = \mu\big( \ell(\cc_i, \bx_k), \ell( \cc_j, \bx_k),\ell(\s, \bx_k)\big), \label{eq:muk} \end{equation}
where $\mu: (0,1)^3 \to \R$,
\begin{align} \mu(x,y,z) & := z \ln \left( \frac{x}{y} \right) +  (1-z)\ln \left( \frac{1-x}{1-y} \right) \label{eq:mufun} \\
& = \ln \left( \dfrac{x^z(1-x)^{(1-z)}}{y^z(1-y)^{(1-z)}} \right). \label{eq:1log}
\end{align}
In this problem, the expected log-likelihood ratio is intimately related to the Kullback-Leibler (KL) divergence $D(P || Q)$, which is a measure of the information lost when using a probability distribution $Q$ to approximate another distribution $P$ \cite[Section 2.1]{burnham2003}. Define
\begin{align}
\K(\s || \x ; \bx_k) & := D\bigg(g( \cdot \mid \s ; \bx_k) \ \bigg| \bigg| \  g(\cdot \mid \x ; \bx_k) \bigg) \label{eq:KLD} \\
& = -\sum_{d=0}^1 g( d \mid \s ; \bx_k) \ln \left[ \frac{g( d \mid \x ; \bx_k)}{g( d \mid \s ; \bx_k)}\right] \nonumber \\
& = -\mu\big( \ell(\x,\bx_k), \ell(\s,\bx_k), \ell(\s,\bx_k) \big), \nonumber
\end{align}
which is the KL divergence \cite[Equation (2.26)]{cover2005} between the true conditional probability distribution for $d_k$, and a distribution which takes $\x \in \R^q$ as the source location. Noting the identity
\begin{equation} \mu(x,z,z) - \mu(y,z,z) = \mu(x,y,z), \label{eq:muprop}\end{equation}
we obtain
\begin{equation} 
\mu^{(i,j)} _ k = \K(\s || \cc_j ; \bx_k) - \K(\s || \cc_i ; \bx_k).
\end{equation}
Now consider a sequence of measurements $\d = (d_1,...,d_n)$ taken at the corresponding locations $\bx_{1:n} = (\bx_1,...,\bx_n)$. The probability distribution of $\d$, given source location $\s$ and measurement locations $\bx_{1:n} \in \R^{nq}$, is
\begin{align} 
G(\d \mid \s ; \bx_{1:n}) & = \prod_{k=1}^n g(d_k \mid \s;\bx_k). \label{eq:p}
\end{align}
The extension of \eqref{eq:KLD} to the distribution of a sequence of measurements is then given by
\begin{align}
\K(\s || \x ; \bx_{1:n}) & := D\bigg(G( \cdot \mid \s ; \bx_{1:n}) \ \bigg| \bigg| \  G(\cdot \mid \x ; \bx_{1:n}) \bigg) \label{eq:KLdef} \\ \label{eq:KLadditivity}
& = \sum_{k=1}^n \K(\s || \x ; \bx_k),
\end{align}
which follows from the additivity property of KL divergence for independent distributions (a corollary of \cite[Theorem 2.5.3]{cover2005}). Thus
\begin{align} \sum_{k=1}^n \mu^{(i,j)}_k  = \K(\s || \cc_j ; \bx_{1:n})  - \K(\s || \cc_i ; \bx_{1:n}). \label{eq:KLDsum} \end{align}

\subsection{General posterior convergence results}
We now present the first theoretical results. We begin by establishing sufficient conditions for \hbox{$\hat{p}_k(i \mid d_{1:k} ; \bx_{1:k})$} to decay to zero at index $i$, as $k \to \infty$. 

\begin{theorem} Let $\s \in S$, and let $(\bx_k)_{k\in\N} \subset \R^q$ be a bounded sequence of measurement locations. If there exists a pair of cells $(i,j)$ and some $p > \frac{1}{2}$ such that 
		\begin{equation} \limsup_{n \to \infty} \left[ \frac{1}{n^p} \sum_{k=1}^{n}\mu^{(i,j)}_k \right] < 0, \label{eq:convergence} \end{equation} then under recursion \eqref{eq:bayesmap}, $\hat{p}_k(i \mid d_{1:k} ; \bx_{1:k}) \to 0$ almost surely as $k \to \infty$.
		\label{thm: estimator}
\end{theorem}
\begin{proof}
	Since we have assumed the co-domain of $\ell$ is $(0,1)$, this implies \begin{equation}\forall d \in \{0,1\}, \ \forall \s,\bx \in \R^q, \ \ g(d \mid \s; \bx) > 0.\label{eq:jpos} \end{equation} Recalling $\hat{p}_0(j) > 0$ for all $j$, recursion \eqref{eq:bayesmap} together with \eqref{eq:jpos} implies that $\hat{p}_k(j\mid d_{1:k} ; \bx_{1:k})>0$ for all $j \in \N_M$ and $k \in \N$. Thus \eqref{eq:pratio} is well defined, and for any pair $(i,j)$, we obtain
	\begin{eqnarray*} \dfrac{\hat{p}_n(i \mid d_{1:n} ; \bx_{1:n})}{\hat{p}_n(j \mid d_{1:n} ; \bx_{1:n})} = \dfrac{\hat{p}_0(i)}{\hat{p}_0(j)} \prod_{k=1}^n Z^{(i,j)}_k. \end{eqnarray*}
	Since $\hat{p}_n(j \mid d_{1:n} ; \bx_{1:n}) \leq 1,$ this yields the inequality
	\begin{equation}
	\hat{p}_n(i \mid d_{1:n} ; \bx_{1:n}) \leq  \dfrac{\hat{p}_0(i)}{\hat{p}_0(j)} \prod_{k=1}^n Z^{(i,j)}_k.\label{eq:ratio}
	\end{equation}
	
	Note that $Z^{(i,j)}_1, Z^{(i,j)}_2,...,$ are independent random variables because our measurements are independent. We also know $\s, \cc_1,...,\cc_M \in S$, where $S \subset \R^q$ is compact. Furthermore, the sequence $\bx_1, \bx_2,...$ is bounded, and therefore never leaves some compact subset $X \subset \R^q$. Since $\ell:\R^{q} \times \R^{q} \to (0,1)$ is continuous, it attains a minimum and maximum on $S \times X$ \cite[Theorem 4.16]{rudin1964-1}. Let \begin{equation} \ell_1 := \max \ell(S, X) < 1 \text{ and } \ell_0 := \min \ell(S, X) > 0. \label{eq:boundedl} \end{equation} Then \eqref{eq:gfun} implies \begin{dmath}\alpha {:=} \min \left\{ \frac{\ell_0}{\ell_1}, \frac{1-\ell_1}{1-\ell_0}  \right\} \leq \dfrac{g(d_k \mid \cc_i; \bx_k)}{g(d_k \mid \cc_j; \bx_k)}  \leq \max \left\{ \frac{\ell_1}{\ell_0}, \frac{1-\ell_0}{1-\ell_1}  \right\}{=:} \beta \label{eq:alpha} \end{dmath}
	for all $i,j,k$. Thus for all $i,j,k$,
	\begin{equation} 0 < \alpha \leq Z^{(i,j)}_k \leq \beta. \label{eq:bounds} \end{equation}
	Now if condition \eqref{eq:convergence} holds for some pair $(i,j)$ and $p > \frac{1}{2}$, then by definition
	$$\limsup_{n \to \infty} \frac{1}{n^p}\sum_{k=1}^n \Ex \left[ \ln Z^{(i,j)}_k \right] <0,$$
	and applying the result of Lemma \ref{lem:infproducts}, $\prod_{k=1}^n Z_k \to 0 \text{ a.s.}$ as $n \to \infty$. Equation \eqref{eq:ratio} then implies ${\hat{p}_n(i \mid d_{1:n} ; \bx_{1:n})} \to 0 \text{ a.s.. }$ 
\end{proof}
\begin{remark}
This result is similar to \cite[Theorem 1]{selvaratnam2017b}, which establishes convergence in probability of ${\hat{p}_k(i \mid d_{1:k} ; \bx_{1:k})}$ for $p = \frac{1}{2}$. Strengthening the requirement to $p > \frac{1}{2}$ allows us to obtain almost sure convergence. 
\end{remark}

Theorem \ref{thm: estimator} provides us with a sufficient condition ensuring that a cell index $i \in \OO$, and this condition requires finding an index $j$ for which $\sum_k \mu^{(i,j)}_k$ diverges at a sufficient rate. 

\begin{remark} Noting \eqref{eq:KLDsum}, if a pair $(i,j)$ satisfy condition \eqref{eq:convergence}, then approximating the source location with $\cc_j$ would asymptotically result in a lower KL divergence from $G(d_1,d_2... \mid \s; \bx_1,\bx_2,...)$ than approximating the source with $\cc_i$.  
\end{remark}

If the true source location coincides with some centre, Theorem \ref{thm: estimator} enables us to state a condition on the measurement location sequence $\bx_1,\bx_2,...$ that guarantees posterior consistency.
\begin{theorem}[Posterior Consistency] \label{th:consistency}
Let $(\bx_k)_{k\in \N} \subset \R^q $ be bounded, and let $\C = \{\cc_1,...,\cc_M\} \subset S$. Suppose $\cc_j = \s$ for some $j \in \N_M$. If $\forall i \neq j,\ \exists p>\frac{1}{2}$ such that
\begin{equation} \liminf_{n \to \infty} \frac{1}{n^p} \sum_{k=1}^n (\ell(\s,\bx_k) - \ell(\cc_i,\bx_k))^2 >0, \label{eq:no2close} \end{equation}
then $\hat{p}_k(j \mid d_{1:k}; \bx_{1:k}) \to 1$ a.s.. 
\end{theorem}
\begin{proof}
	For any $\x \in \R^q$, the total variation distance between distributions $g( \cdot | \s, \bx_k)$ and $g( \cdot | \x, \bx_k)$ is given by
	\begin{align*} \sup_{d \in \{0,1\}} |g( d| \s, \bx_k) - g( d | \x, \bx_k)| 
	= |\ell(\s,\bx_k) - \ell(\x,\bx_k)|.
	\end{align*}
	Pinsker's inequality \cite[Lemma 2.5]{tsybakov2009} is a lower bound on KL divergence, which then yields
	$$ \K(\s||\x ; \bx_k) \geq 2(\ell(\s,\bx_k) - \ell(\x,\bx_k))^2.$$
	By assumption $\cc_j = \s$, which implies $\K(\s||\cc_j ; \bx_k) = 0$ for any $\bx_k$. Thus, \eqref{eq:KLadditivity} and \eqref{eq:KLDsum} imply that for all $n \in \N$ and $i \in \N_M$,
	$$ \sum_{k=1}^n \mu^{(i,j)}_k  =  - \K(\s||\cc_i ; \bx_{1:n}) \leq - \sum_{k=1}^n 2(\ell(\s,\bx_k) - \ell(\cc_i,\bx_k))^2,$$ which in turn implies
	\begin{align*} \limsup_{n \to \infty} \frac{1}{n^p} \sum_{k=1}^{n}\mu^{(i,j)}_k & \leq \limsup_{n \to \infty} \frac{-2}{n^p} \sum_{k=1}^n (\ell(\s,\bx_k) - \ell(\cc_i,\bx_k))^2 \\
	& = -\liminf_{n \to \infty} \frac{2}{n^p} \sum_{k=1}^n (\ell(\s,\bx_k) - \ell(\cc_i,\bx_k))^2. \end{align*}
	If for all $i \neq j$, there exists $p > \frac{1}{2}$ such that \eqref{eq:no2close} holds, then the LHS of the above inequality is strictly negative. Theorem \ref{thm: estimator} then implies $\hat{p}_k(i \mid d_{1:k}; \bx_{1:k}) \to 0$ a.s. for all $i \neq j$. Since $\hat{p}_k$ is a probability distribution, this implies $\hat{p}_k(j \mid d_{1:k}; \bx_{1:k}) \to 1$ a.s.. 
\end{proof}
\begin{remark} \label{rem:indistinguishable}
If $ \ell(\cc_i,\bx_k) = \ell(\s,\bx_k)$, then having the source at $\cc_i$ yields the same probability of detection at $\bx_k$, as if the source was at $\s$. Thus, $\cc_i$ cannot be distinguished from $\s$ using measurements taken at $\bx_k$. Condition \eqref{eq:no2close} ensures the  agents take readings sufficiently often at locations which provide enough information to distinguish between cells.
\end{remark}
Next, we consider what happens when the source does not coincide with any centre.  
\subsection{Posterior convergence under periodic location sequences} \label{sec:nperiodic}
Analysing the general case where $\s \notin \C$ is difficult when considering completely arbitrary agent location sequences. We therefore restrict our attention to those that are periodic. Many bounded sequences of practical interest are either periodic, or converge to one that is. Furthermore, the effectiveness of any finite sequence $\bx_1,...,\bx_n$ can be analysed by considering a periodic sequence for which $\bx_1,...,\bx_n$ constitutes a single period. 

\begin{definition} A sequence $(\bx_k)_{k\in\N}$ is $n$-periodic for $n \in \N$  iff $\bx_k = \bx_{k+n}$ for all $k$. 
\end{definition} 
Examples of this include \be[i)]
\item a single agent moving in a periodic trajectory, taking measurements at the same locations every $n$ time-steps
\item a team of $n$ agents remaining stationary, and taking measurements in a fixed order.
\end{enumerate}
Any $n$-periodic location sequence is fully specified by the vector $\bx_{1:n} = (\bx_1,...,\bx_n) \in \R^{nq}$. 
For $n$-periodic trajectories, $\sum_{k=1}^n \mu^{(i,j)}_k < 0$ is a sufficient condition for \eqref{eq:convergence}. 
\begin{remark}
	The assumption of $n$-periodicity can be weakened. Partition an arbitrary sequence $\bx_1,\bx_2,...$ into blocks of some fixed length $L \in \N$, and permute the measurement order within these blocks. This has no effect on condition \eqref{eq:convergence}. Thus, if there exists such permutation which yields an $n$-periodic sequence, the results of this section (and Section \ref{sec:lambdarho}) apply without modification to the original sequence. 
\end{remark}
Equation \eqref{eq:KLDsum} then implies that, in the limit, the algorithm selects the indices of centres which, when treated as the source location, minimise KL divergence from the true measurement probability distribution. This is stated precisely below. 
\begin{theorem}
	\label{th:general}
	Let $\s \in S$ and let $(\bx_k)_{k\in \N} \subset \R^q$ be $n$-periodic for some $n \in \mathbb{N}$. Then \begin{equation} \OO^c \subset \argmin_{i \in \N_M} \K(\s||\cc_i ; \bx_{1:n} ) =: \B. \label{eq:setB} \end{equation}
\end{theorem}
\begin{proof}
	Consider any index $i \notin \B$. By definition, for any $j \in \B$, we have
	$$ \K(\s||\cc_i ; \bx_{1:n} ) > \K(\s||\cc_j ; \bx_{1:n} ),$$ 
	which implies $\sum_{k=1}^n \mu^{(i,j)}_k < 0$ by \eqref{eq:KLDsum}. 
	Since $(\bx_k)_{k \in \N}$ is $n$-periodic, $$ \forall m \in \N, \ \sum_{k=1}^{m} \mu^{(i,j)}_k = \left\lfloor \frac{m}{n} \right\rfloor \sum_{k=1}^n\mu^{(i,j)}_k +   \sum_{k=1}^{m \text{ mod } n} \mu^{(i,j)}_k.$$ Note that $\lim_{m \to \infty} m^{-1} \lfloor \frac{m}{n} \rfloor = \frac{1}{n}$ and $(m \text{ mod }n) < n$ for any $m$. Therefore, \begin{align*} \lim_{m \to \infty} & \frac{1}{m} \sum_{k=1}^{m} \mu^{(i,j)}_k  \\
	& = \lim_{m \to \infty}\left( \frac{\lfloor \frac{m}{n} \rfloor }{m}\sum_{k=1}^n \mu^{(i,j)}_k + \frac{1}{m} \sum_{k=1}^{m \text{ mod } n} \mu^{(i,j)}_k  \right) \\
	& = \frac{1}{n} \sum_{k=1}^{n} \mu^{(i,j)}_k < 0. \end{align*} Thus \eqref{eq:convergence} is satisfied using $p = 1$, and $ {\hat{p}_k(i \mid d_{1:k} ; \bx_{1:k})} \to 0$ a.s. by Theorem \ref{thm: estimator}, which by definition implies $i \in \OO$. We have shown $\B^c \subset \OO$, which is equivalent to $\OO^c \subset \B$. 
\end{proof}
\begin{remark} \label{rem:AnB}
Equation \eqref{eq:setB} reveals that the posterior may fail to decay to zero only at indices which minimise ${\K(\s||\cc_i ; \bx_{1:n})}$. The centres corresponding to these indices are solutions to $ {\argmin_{\x \in \mathcal{C}} \K(\s||\x ; \bx_{1:n} )}$. They can be considered approximate solutions to 
\begin{equation} \argmin_{\x \in S} \K(\s||\x ; \bx_{1:n} ),
\label{eq:contopt} \end{equation}
where the optimisation now takes place over the entire search region instead of the finite set $\C$.
\end{remark}
The problem \eqref{eq:contopt} is more amenable to analysis, as it does not depend on the choice of centres in $\C$, and its solution set provides us with additional insight.
\begin{restatable}{proposition}{sinJmax}
\label{prop:sinJmax}
Given $\bx_1,...,\bx_n \in \R^q$ and $\s \in S$, $${\min_{\x \in S} \K(\s||\x ; \bx_{1:n} ) = 0}, \text{ and }$$ \begin{equation} \argmin_{\x \in S} \K(\s||\x ; \bx_{1:n} ) = \bigcap_{k=1}^n \{ \x \in S \mid \ell(\x, \bx_k) = \ell(\s, \bx_k) \} := \A. \label{eq:seteq} \end{equation} Therefore, 
$\s \in \displaystyle \argmin_{\x \in S} \K(\s||\x ;\bx_{1:n}) $. 
\end{restatable}
\begin{proof}
	Here we exploit the properties of KL divergence stated in \cite[Theorem 2.6.3]{cover2005}. 
	Since KL divergence is non-negative, we have ${\K(\s||\x ; \bx_{1:n}) \geq 0}$ for all $\x$. From \eqref{eq:p},
	$$ p(\d \mid \s;\bx_{1:n}) = \prod_{k=1}^n \ell(\s,\bx_k)^{d_k} \left[ 1 - \ell(\s,\bx_k) \right]^{1 - d_k}.$$
	If $\x \in \A$, then $\ell(\x,\bx_k) = \ell(\s,\bx_k)$ for all $k \in \N_n$, which implies ${p(\d \mid \s;\bx_{1:n})} = {p(\d \mid \x;\bx_{1:n})}$ for all $\d \in \{0,1\}^n$. Recalling \eqref{eq:KLdef}, this implies ${\K(\s || \x; \bx_{1:n}) = 0}$. 
	
	Recalling \eqref{eq:KLadditivity}, if ${\K(\s||\x;\bx_{1:n}) = 0}$, then ${\K(\s||\x;\bx_k) = 0}$ for all $k \in \N_n$. According to the definition \eqref{eq:KLD}, this holds if and only if \begin{equation*} \forall k \in \N_n,\ \forall d \in \{0,1\},\ g(d \mid \s,\bx_k) =g(d\mid \x ; \bx_k). \label{eq:gequal} \end{equation*} Finally referring to \eqref{eq:gfun}, choosing $d = 1$ implies $\ell(\s,\bx_k) = \ell(\x,\bx_k)$ for all $k \in \N_n$, and therefore $\x \in \A$. 
\end{proof}
 Thus \eqref{eq:contopt} contains only the candidate locations that are indistinguishable from the source based on the entire history of measurements (see Remark \ref{rem:indistinguishable}). Given this characterisation, it is obviously desirable to define the $n$-periodic sequence $\bx_{k \in \N}$ by choosing $\bx_1,...,\bx_n$ such that $\A = \{\s\}$. This is a useful requirement to impose when planning agent trajectories, as it guarantees there is sufficient information available from the measurements to uniquely identify the source. 
\begin{remark} \label{rem:mild}
	Observe that $\A \subset \R^q$ is the solution set of $n$ simultaneous non-linear equations.  Therefore if $n > q$, typically only mild conditions on the measurement location geometry are required to guarantee $\A = \{\s\}$. Such conditions are developed in Section \ref{sec:rho} for the case in which the probability of detection is purely a function of distance from the source. 
\end{remark}
 
If $\A = \{\s\}$ and $\s \in \C$, the posterior is consistent and estimation algorithm will eventually unambiguously identify the source index.
\begin{corollary} \label{cor:unambiguous}
Let $(\bx_k)_{k\in \N} \subset \R^q $ be $n$-periodic, and let $\C = \{\cc_1,...,\cc_M\} \subset S$. Suppose $\cc_j = \s$ for some $j \in \N_M$. If  $\A = \{\s\}$, then ${\hat{p}_k(j \mid d_{1:k}; \bx_{1:k})} \to 1$ a.s.. 
\end{corollary}
\begin{proof}
	Suppose $\A = \{\s\} = \{\cc_j\} = \argmin_{\x \in S}\K(\s||\x ; \bx_{1:n})$. Since $\cc_j \in \C \subset S$, we have $\B = \{\cc_j\}$. The centres are distinct and thus for all $i \neq j$, $\cc_i \notin \B$, which implies $i \in \OO$ by Theorem \ref{th:general}. Since $\hat{p}_k$ is a probability distribution, this implies ${\hat{p}_k(j\mid d_{1:k}; \bx_{1:k})} \to 1$ a.s..
\end{proof}
If $\A = \{\s\}$ but $\s$ does not coincide with the centres, the support of the posterior (in the limit) will contain approximations to $\s$ in $\C$ that yield the lowest KL divergence from $G(\cdot \mid \s,\bx_{1:n})$. 
\begin{remark}
Note that cells having the lowest $\K$ value do not necessarily coincide with the centres closest to the source: 
$$ \| \s - \cc_j \| \leq \| \s - \cc_i\| \notimplies \K(\s||\cc_j ; \bx_{1:n}) \leq \K(\s||\cc_i ; \bx_{1:n}).$$ \end{remark} According to the remark above, if the cells are based on the Voronoi decomposition of $S$, there is no guarantee of consistency (with respect to Definition \ref{def:consistency}) when $\s \notin \C$. However, simulation results in Section \ref{sec:sim} indicate that the estimation procedure still performs well when the centres are chosen in a grid.

\subsection{Range dependent probability of detection} \label{sec:rho}
In many scenarios, the probability of detection will be purely a function of the distance from the source to an agent. We can then say $\ell$ is of the form $ \ell(\x, \bx) = \rho(\|\x - \bx \|)$, where \hbox{$\rho: [0, \infty) \to (0,1)$}. Typically $\rho$ will be strictly decreasing, and therefore injective. We now consider the application of the previous results to this case.

We begin by examining the consistency requirement \eqref{eq:no2close} for a general measurement location sequence in Theorem \ref{th:consistency}, when the source coincides with some centre.
\begin{lemma} \label{lem:continverse}
	Let $\s \in S$, let $(\bx_k)_{k \in \N} \subset \R^q$ be a bounded sequence, and assume $\ell(\x, \bx) = \rho(\|\x - \bx \|)$, where \hbox{$\rho: [0, \infty) \to (0,1)$} is continuous and injective. Then
	\begin{multline*}
	\forall \delta>0,\ \exists \epsilon \in (0,1),\\ \big| \| \bx_k - \s \| - \| \bx_k - \cc_i \| \big| \geq \delta \implies | \ell(\s,\bx_k) - \ell(\cc_i,\bx_k) | \geq \epsilon.
	\end{multline*}
\end{lemma}
\begin{proof}
	Let $z := \sup \{ \| \bx_k - \cc \| \mid k \in \N,\ \cc \in \C \cup \{\s\} \} < \infty,$ and let $D := \rho([0,z])$. Let $\bar{\rho}:[0,z] \to D,\ \bar{\rho}(r) = \rho(r)$, which is a bijection because $\rho$ is injective. By \cite[Theorem 4.17]{rudin1964-1}, $\bar{\rho}^{-1}$ is continuous and therefore,
	\begin{dmath} {\forall l_1,l_2 \in D,\ \forall \delta> 0,\ \exists \epsilon> 0,}\\ { |l_1 - l_2 | < \epsilon \implies |\bar{\rho}^{-1}(l_1) - \bar{\rho}^{-1}(l_2)| < \delta.}\label{eq:rudin} \end{dmath}
	Now for any $r_i \in [0,z]$, $l_i = \bar{\rho}(r_i) \in D$. Therefore \eqref{eq:rudin} implies
	\begin{dmath*} {\forall r_1,r_2 \in [0,z],\ \forall \delta> 0,\ \exists \epsilon> 0,}\\ {|\bar{\rho}(r_1) - \bar{\rho}(r_2) | < \epsilon \implies |r_1 - r_2| < \delta,} \end{dmath*}
	which is in turn equivalent to
	\begin{dmath}{\forall r_1,r_2 \in [0,z],\ \forall \delta> 0,\ \exists \epsilon> 0,}\\ {|r_1 - r_2| \geq \delta \implies  |\bar{\rho}(r_1) - \bar{\rho}(r_2) | \geq \epsilon.} \label{eq:almost} \end{dmath}
	The image of $\bar{\rho}
	$ is contained in $(0,1)$, which implies $\epsilon \in (0,1)$. 
	Now $\| \bx_k - \s\|,\|\bx_k - \cc_i\| \in [0,z]$ by definition of $z$, and therefore $\bar{\rho}(\| \bx_k - \s\|) = \ell(\s,\bx_k) $ and $\bar{\rho}(\| \bx_k - \cc_i\|) = \ell(\cc_i,\bx_k).$ Thus choosing $r_1 = \| \bx_k - \s\|$ and $r_2 = \|\bx_k - \cc_i\|$ in \eqref{eq:almost} completes the proof. 
\end{proof}
\begin{remark} \label{rem:straightline}
	Referring to the Lemma above, condition \eqref{eq:no2close} is met for a particular cell $i$ if $\big| \| \bx_k - \s \| - \| \bx_k - \cc_i \| \big| \geq \delta$ occurs sufficiently often for the same $\delta > 0$. A simple way to guarantee this holds for every $\cc_i \neq \s$ is
	to make sure the location sequence does not travel in (or converge to) a straight line indefinitely.
\end{remark}

We now turn our attention to periodic locations sequences, allowing $\s \in S$ to be arbitrary. 
As discussed in Section \ref{sec:nperiodic}, the basic requirement for an $n$-periodic location sequence is to ensure $\A = \{\s\}$.
If enough readings are taken at location $ \bx$, then the probability of detection $\ell(\s, \bx)$ can be estimated from the ratio of hits to misses. When $\rho$ is injective, this probability of detection can be mapped back to place $\bx$ at a unique distance from $\s$. This suggests a strategy akin to trilateration will ensure $\s$ is the unique solution to \eqref{eq:contopt}, and the result below confirms our intuition.
Let $\mathrm{aff}(X)$ denote the affine hull of some $X \subset \R^q$.  
\begin{proposition}
	\label{prop:s=Jmon}
	Let $\s \in S \subset \R^q$, where $q \in \{2,3\}$. Let $\bx_1,...,\bx_n \in S$ and let $\ell$ be of the form $ \ell(\x, \bx) = \rho(\|\x - \bx \|)$, where \hbox{$\rho: [0, \infty) \to (0,1)$} is continuous and injective.
	If $\dim \mathrm{aff}(\{\bx_1,...,\bx_n\}) = q$, then $$\A := \bigcap_{k=1}^n \{ \x \in S \mid \ell(\x, \bx_k) = \ell(\s, \bx_k) \} = \{\s\}.$$
\end{proposition}
\begin{proof}
	Since $\rho$ is injective,
	\begin{align*} \A & = \bigcap_{k=1}^n \{ \x \in S \mid \rho(\| \x- \bx_k \|) = \rho( \| \s -  \bx_k \| ) \} \\
	& = \bigcap_{k=1}^n \{ \x \in S \mid \| \x- \bx_k \| =  \| \s -  \bx_k \| \}, \end{align*}
	which is the intersection of $n$ spheres in $S \subset \R^q$. For their intersection to be unique, it is sufficient for three of the $\bx_k$ to not be collinear when $q=2$, and four of the $\bx_k$ to not be coplanar when $q = 3$.
\end{proof}
\section{D-Optimal measurement locations} \label{sec:fish} 
As noted in Remark \ref{rem:mild}, the requirement that $\A = \{\s\}$ typically imposes only mild constraints on the geometry of an $n$-periodic measurement location sequence if $n > q$. While this guarantees there is sufficient information available to uniquely identify the source, it makes no claim to optimality. In this section, we first attempt to optimise the measurement locations with respect to the determinant of the Bayesian Information Matrix (BIM) \cite[Section 4.3.3.2]{vantrees2013}. The inverse of the BIM is the Bayesian Cramer-Rao bound, a lower bound on the MSE of any estimator\footnote{subject to the bias conditions in \cite[Equation (4.522)]{vantrees2013}.}. Thus, maximising the BIM determinant minimises a lower bound on the volume of the estimator's concentration ellipsoids \cite[Section 4.3.2.1]{vantrees2013}. 

Recall we have a team of $N$ agents, and suppose they each report one measurement in a fixed sequence every $N$ time-steps.
Consider the following question: given the information $(\bx_k,d_k)_{k \in \N_n}$ received up to time $n$, what are the D-optimal agent locations $\bX := (\bx^+_1,...,\bx^+_N) \in \R^{Nq}$ at which to take the next $N$ measurements $\d^+ := (d^+_1,...,d^+_N) \in \{0,1\}^N$?

The joint probability distribution for $\d^+$ and $\s$ conditioned on the information received is given by
\begin{dmath*} {p(\d^+, \s \mid d_{1:n} ; \bx_{1:n}, \bX)} = {G(\d^+ \mid \s ; \bX)p_n(\s \mid d_{1:N}; \bx_{1:N}),} \end{dmath*}
where $G$ is defined in \eqref{eq:p}.
The Bayesian information matrix for estimating $\s$ from $\d^+$, given the information received up to time $n$, is then
$$ \J_n(\bX) : = - \Ex \left[ \nabla^2_\s \ln p(\d^+, \s \mid d_{1:n} ; \bx_{1:n}, \bX) \mid d_{1:n}  \right].$$ Let $\mathcal{X}_n \subset \R^{Nq}$ denote the feasible set of agent locations at time $n$ (which may, for example, incorporate motion constraints). Solving
\begin{equation} \bx_{n+1:n+N} = \argmax_{\bX \in \mathcal{X}_n} \det \J_n(\bX ) \label{eq:fullproblem} \end{equation} yields D-optimal locations for the next $N$ measurements.
Although \cite{zheng2012,zuo2011a} provide a recursive method for computing the BIM numerically, the problem \eqref{eq:fullproblem} is non-convex and obtaining a direct solution to it is intractable. Instead, we propose a relaxed version of the problem that optimises the classical Fisher Information Matrix (FIM). Analytical solutions to the relaxed problem can be derived under additional assumptions.  
The BIM can be written as 
\begin{dmath} {\J_n(\bX)} = {\Ex\left[ \bI(\s ; \bX) \mid d_{1:n} \right]} + {\J^{\ominus}_n(d_{1:n},\bx_{1:n}),} \label{eq:BIMdecomposition} \end{dmath}
where
$$ \J^{\ominus}_n(d_{1:n},\bx_{1:n}) := -\Ex \left[ \nabla^2_\s \ln p_n(\s \mid d_{1:n}; \bx_{1:n}) \mid d_{1:n} \right]$$
is the contribution of the information information already received, and
\begin{align*} 
\bI(\s ; \bX) & = - \Ex \left[ \nabla_\s^2 \  \ln G( \d^+ \mid \s; \bX) \mid \s \right]  \\
& = - \Ex \left[ \nabla_\s^2 \sum_{k=1}^{N} \ln g(d^+_k \mid \s;\bx^+_k) \mid \s \right] \\ 
& = - \sum_{k=1}^{N} \Ex \left[ \nabla_\s^2  \ln g(d^+_k \mid \s;\bx^+_k) \mid \s \right].
\end{align*}
is the classical FIM conditioned on the source location \cite[Equation (4.515)]{vantrees2013}.
Note that both $\J^{\ominus}_n(d_{1:n},\bx_{1:n}) $ and $ \bI(\s ; \bX)$ are symmetric positive semi-definite \cite[(4.519), (4.393)]{vantrees2013}. It then follows from the Minkowski determinant inequality that
\begin{align} \det \J_n(\bX ) & \geq \det \Ex\left[ \bI(\s ; \bX) \mid d_{1:n}\right] + \det \J^{\ominus}_n(d_{1:n},\bx_{1:n}) \nonumber \\
 & \approx \det  \bI(\hat{\s}_n ; \bX) + \det \J^{\ominus}_n(d_{1:n},\bx_{1:n}), \label{eq:approxdecomposition}
\end{align}
where
$$ \hat{\s}_n := \Ex[ \s \mid d_{1:n}]  $$ 
represents the mean of the posterior $p_n$. Since the second term of \eqref{eq:approxdecomposition} does not depend on $\bX$, 
\begin{equation} \bx_{n+1:n+N} = \argmax_{\bX \in \mathcal{X}_n} \det \bI(\hat{\s}_n;\bX),
\label{eq:optimalsensor} \end{equation}
is a relaxation of \eqref{eq:fullproblem} that chooses measurement locations to maximise the determinant of the FIM evaluated at the expected source location. The FIM has the structure
\begin{equation} \bI(\s;\bX) = \sum_{k=1}^{N} \I(\s;\bx^+_k), \label{eq:totalfish} \end{equation}
where 
\begin{align}
\I(\s;\bx_k) &:= - \Ex \left[ \nabla_\s^2 \ln g(d_k\mid \s;\bx_k) \right] \\
& = \dfrac{\nabla_\s \ell(\s,\bx_k)\nabla_\s \ell(\s,\bx_k)^\top}{\ell(\s,\bx_k)[1 - \ell(\s,\bx_k)]}. \label{eq:fishi}
\end{align} is the FIM for a single reading taken at $\bx_k$. Equation \eqref{eq:fishi} is derived in Appendix \ref{app:Fish}. We exploit this structure below to obtain an analytic solution to \eqref{eq:optimalsensor} for a range-dependent probability of detection, under particular distance constraints. 
\subsection{Range dependent probability of detection}
An expression for the FIM as a function of the source and agent locations is derived below, assuming the probability of detection is a smooth function of distance. We focus on localisation in the plane, letting $\s = (s_1,s_2)$ and $\bx_k = (\xi_{k,1}, \xi_{k,2})$.
\begin{proposition}[Fisher Information Matrix] \label{prop:fish} Let $ \s \in S \subset \R^2$, and assume $\ell(\s,\bx) = \rho(\|\s - \bx\|)$, where $\rho:[0,\infty) \to (0,1)$ is continuously differentiable. 
	Let $\bx_{1:N} \in \R^{2N}$ be such that, $\forall k \in \N_N,\ \bx_k \neq \s$. Define $$r_k := \|\s - \bx_k\|, \quad \theta_k := \mathrm{atan2} \big( \xi_{k,2} - s_2, \xi_{k,1} - s_1 \big).$$ Then
	\begin{equation}
	\bI(\s;\bx_{1:N}) = \sum_{k=1}^N \dfrac{ \rho'(r_k)^2}{\rho(r_k)[1-\rho(r_k)]} \begin{bmatrix} \cos^2(\theta_k) & \frac{\sin(2\theta_k)}{2} \\ \frac{\sin(2\theta_k)}{2} & \sin^2(\theta_k) \end{bmatrix}. \label{eq:distancefish}
	\end{equation}
\end{proposition}
\begin{proof}	Since $\ell(\s,\bx_k) = \rho(r_k)$, it follows from the chain rule that
	$$ \nabla_\s \ell(\s,\bx_k) = \rho'(r_k) \nabla_\s r_k. $$
	Equation \eqref{eq:fishi} then implies
	\begin{equation*} \I(\s;\bx_k) = \dfrac{ \rho'(r_k)^2}{\rho(r_k)[1-\rho(r_k)]} \nabla_\s r_k \nabla_\s r_k^\top 
	.\end{equation*}
	Now $ \nabla_\s r_k = \dfrac{\s - \bx_k}{r_k}$, and note $\bx_k \neq \s \iff r_k > 0$. By definition of $\theta_k$, \\ $\bx_k - \s = r_k \begin{bmatrix} \cos \theta_k \\ \sin \theta_k \end{bmatrix}$, and therefore $\nabla_\s r_k = -\begin{bmatrix} \cos \theta_k \\ \sin \theta_k \end{bmatrix}$ for $r_k>0$. Thus
	\begin{equation} \I(\s;\bx_k) = \dfrac{ \rho'(r_k)^2}{\rho(r_k)[1-\rho(r_k)]} \begin{bmatrix} \cos^2(\theta_k) & \frac{\sin(2\theta_k)}{2} \\ \frac{\sin(2\theta_k)}{2} & \sin^2(\theta_k) \end{bmatrix} \label{eq:1fish}
	,\end{equation} 
	and applying \eqref{eq:totalfish} yields \eqref{eq:distancefish}.
\end{proof}
We now derive conditions for an optimal geometry, under the constraint that the agents be equidistant from the source location. 
Letting $r := r_1 = \hdots =r_N > 0$, the FIM becomes 
\begin{equation}
\bI(\s;\bx_{1:N}) =  \dfrac{ \rho'(r)^2}{\rho(r)[1-\rho(r)]} \sum_{k=1}^n \begin{bmatrix} \cos^2(\theta_k) & \frac{\sin(2\theta_k)}{2} \\ \frac{\sin(2\theta_k)}{2} & \sin^2(\theta_k) \end{bmatrix}. \label{eq:equifish}
\end{equation}
It is clear from \eqref{eq:equifish} that the optimal angles and radius can now be chosen independently.
\begin{theorem}[Optimal Sensor Geometry]
	\label{thm:OSG}
	Assume $\s \in S \subset \R^2$, and let $\ell(\s,\bx) = \rho(\|\s - \bx\|)$, where \hbox{$\rho:[0,\infty) \to (0,1)$} is continuously differentiable. Constrain \\ \hbox{$\bx_{1:N} \in \R^{2N}$} to be such that \begin{equation} \forall k,m \in \N_N, \ \| \bx_k - \s \| = \| \bx_m - \s \| > 0. \label{eq:equidist} \end{equation} Define $\theta_k := \mathrm{atan2} \big( \xi_{k,2} - s_2, \xi_{k,1} - s_1 \big)$ and $r := \| \bx_1 - \s \|$. 
	\be
	\item \label{pt:Xr} For any fixed $r > 0, \ \det \bI(\s;\bx_{1:N})$ is maximised if and only if
	\begin{equation} \sum_{k=1}^N \cos( 2\theta_k) = 0 \text{ and } \sum_{k=1}^N \sin( 2\theta_k) = 0. \label{eq:optimalangles} \end{equation}
	\item \label{pt:Xrr} Apply the additional constraint $r \in [r_1,r_2]$ for some $0 < r_1 \leq r_2$. Then for any fixed $\theta_1,...,\theta_N,\ \det \bI(\s;\bx_{1:N})$ is maximised
	if and only if \begin{equation} r \in \argmax_{x \in [r_1,r_2]} \dfrac{ \rho'(x)^2}{\rho(x)[1-\rho(x)]}. \label{eq:optimalr} \end{equation}
	\item \label{pt:Xrtheta}	Optimising jointly over $\theta_1,...,\theta_N \in \R$ and $r \in [r_1,r_2]$, $\det \bI(\s;\bx_{1:N})$ is maximised if and only if \eqref{eq:optimalangles} and \eqref{eq:optimalr} both hold. 
	\ne
\end{theorem} 
\begin{proof}
	Under the constraint \eqref{eq:equidist}, 
	\begin{dgroup*} 
		\begin{dmath*} \det \bI(\s;\bx_{1:N})  =\dfrac{ \rho'(r)^2}{\rho(r)[1-\rho(r)]} \det \left( \sum_{k=1}^N  \begin{bmatrix} \cos^2(\theta_k) & \frac{\sin(2\theta_k)}{2} \\ \frac{\sin(2\theta_k)}{2} & \sin^2(\theta_k) \end{bmatrix}  \right) \end{dmath*}
		\begin{dmath} = \dfrac{ \rho'(r)^2}{\rho(r)[1-\rho(r)]} \det \left( \sum_{k=1}^N  \begin{bmatrix} \cos^2(\theta_k) & \frac{\sin(2\theta_k)}{2} \\ \frac{\sin(2\theta_k)}{2} & \sin^2(\theta_k) \end{bmatrix}^\top \right), \label{eq:detI} \end{dmath}
	\end{dgroup*}
	which has a form identical to the Fisher Information determinant for range-only measurements \cite[Equation (13)]{bishop2010}. It is also implied by \cite[Theorem 2]{bishop2010} that  $ \det \left( \displaystyle \sum_{k=1}^N  \begin{bmatrix} \cos^2(\theta_k) & \frac{\sin(2\theta_k)}{2} \\ \frac{\sin(2\theta_k)}{2} & \sin^2(\theta_k) \end{bmatrix}^\top \right)$ is maximised if and only if \eqref{eq:optimalangles} is satisfied, and this proves Statement \ref{pt:Xr}. The coefficient $\dfrac{ \rho'(r)^2}{\rho(r)[1-\rho(r)]}$ is a continuous function of $r$, and this guarantees the existence of a maximum on $[r_1,r_2]$. Statements \ref{pt:Xrr} and \ref{pt:Xrtheta} then follow immediately from \eqref{eq:detI}. 
\end{proof}
\begin{remark}
	A particular type of geometry that satisfies condition \eqref{eq:optimalangles} is to have the agents spaced out at equal angles about the source. This result is stated in \cite[Proposition 2]{bishop2010}. Other types geometries satisfying \eqref{eq:optimalangles} can also be found in the same work. 
\end{remark}
The above result identifies optimal measurement locations with respect to the Fisher Information determinant, given the source location $\s$. In practice, of course, $\s$ is unknown. However, as emphasized at the beginning of Section \ref{sec:fish}, an approximate solution to \eqref{eq:fullproblem} can be generated by optimising the Fisher Information determinant evaluated at the expected source location.
For some $0 < r_1 \leq r_2$, let the feasible set at time $n$ be $$\mathcal{X}_n = \{ \bx_{1:N} \mid \forall k,m \in \N_N,\ \| \bx_k - \hat{\s}_n \| = \| \bx_m - \hat{\s}_n \| \in [r_1,r_2] \}.$$
It follows directly from Theorem \ref{thm:OSG} that an exact solution to \eqref{eq:optimalsensor} is given by
\begin{equation} \forall k \in \N_N,\ \bx_{n+k} =  \hat{\s}_n + r \begin{bmatrix} \cos\theta_k \\ \sin\theta_k\end{bmatrix}, \label{eq:exactsoln} \end{equation}
where $r,\theta_1,...,\theta_N$ satisfy \eqref{eq:optimalangles} - \eqref{eq:optimalr}.

\section{Inexact knowledge of probability of detection} \label{sec:lambda}
In practice, the probability-of-detection function $\ell$ will not be known completely. Suppose that, instead, we have knowledge of a continuous function $\hat{\ell}:\R^q \times \R^q \to (0,1)$ that satisfies
\begin{equation} \forall \s,\x \in \R^q,\ \hat{\ell}(\s,\x) \geq \ell(\s,\x). \label{eq:envelope} \end{equation}
That is, $\hat{\ell}$ is an envelope for $\ell$. 
If the algorithm uses $\hat{\ell}$ in place of $\ell$ when computing \eqref{eq:bayesmap}, the corresponding version of the likelihood function is
\begin{equation} \hat{g}(d_k \mid \s; \bx_k) = \hat{\ell}(\s,\bx_k)^{d_k} \left[ 1 - \hat{\ell}(\s,\bx_k) \right]^{1 - d_k}. \end{equation}
By the same argument as \eqref{eq:pratio}, it is clear the convergence of the posterior depends on the ratio
 \begin{equation} \hat{Z}^{(i,j)}_k :=  \frac{ \hat{g}(d_k \mid \cc_i; \bx_k)}{ \hat{g}(d_k \mid \cc_j ; \bx_k)}, \end{equation} the logarithm of which has expected value
 \begin{align} \hat{\mu}^{(i,j)} _ k & = \Ex \left[ \ln \hat{Z}^{(i,j)}_k \mid \s \right] \\
& = \mu\big( \hat{ \ell}(\cc_i, \bx_k), \hat{\ell}( \cc_j, \bx_k),\ell(\s, \bx_k)\big).\end{align}
We emphasize that this expectation is taken with respect to the true distribution $g(\cdot \mid \s;\bx_k)$, and remind the reader that $\mu$ is defined in \eqref{eq:mufun}. The KL divergence originally defined in \eqref{eq:KLdef} now generalizes to
\begin{align} \K(\s, \ell \  || \  \x, \hat{\ell}  ; \bx_{1:n}) &:= D\bigg(G( \cdot \mid \s ; \bx_{1:n}) \ \bigg| \bigg| \  \hat{G}(\cdot \mid \x ; \bx_{1:n}) \bigg) \\
& = - \sum_{k=1}^n \mu\big( \hat{\ell}(\x, \bx_k), \ell( \s, \bx_k),\ell(\s, \bx_k)\big),
\label{eq:KLsum2} \end{align}
where $\hat{G}(\d \mid \s ; \bx_{1:n})  = \prod_{k=1}^n \hat{g}(d_k \mid \s;\bx_k)$.
 Using property \eqref{eq:muprop}, we then obtain
\begin{align}
\sum_{k=1}^n \hat{\mu}^{(i,j)}_k & = \sum_{k=1}^n \mu\big( \hat{\ell}(\cc_i, \bx_k), \hat{\ell}( \cc_j, \bx_k),\ell(\s, \bx_k)\big) \\
&  = \K(\s, \ell \  || \  \cc_j, \hat{\ell}  ; \bx_{1:n})  - \K(\s, \ell \  || \  \cc_i, \hat{\ell}  ; \bx_{1:n}).
\label{eq:KLdiff}
\end{align}
\begin{remark} \label{rem:same}
When $\hat{\ell} \neq \ell$, the convergence result in Theorem \ref{thm: estimator} holds when $\mu^{(i,j)} _ k$ is replaced with $\hat{\mu}^{(i,j)} _ k$. Similarly, Theorem \ref{th:general} holds when $\K(\s||\x ;\bx_{1:n})$ is replaced with $\K(\s, \ell \  || \  \x, \hat{\ell}  ; \bx_{1:n})$.\end{remark}
\subsection{Periodic Measurement Location Sequences} \label{sec:lambdarho}
Once again, we restrict attention to $n$-periodic agent location sequences. In general, according to Theorem \ref{th:general} and Remark \ref{rem:same}, the posterior will decay to zero at every index outside the set
\begin{equation} \B(\hat{\ell}\mid \ell) := \argmin_{i \in \N_m} \K(\s, \ell \  || \  \x, \hat{\ell}  ; \bx_{1:n}). \label{eq:Bdef2} \end{equation}
Consider two cases. In both cases, the algorithm is run assuming the same envelope $\hat{\ell}$. The first case is a special case in which $\ell = \hat{\ell}$, and in the second case $\ell$ is arbitrary. Asymptotically, the support of the posterior is then contained in $ \B(\hat{\ell}\mid\hat{\ell})$ and $\B(\ell\mid\hat{\ell})$ respectively. 
We now compare these two sets.

\begin{lemma} Suppose that, \begin{equation} \forall k \in \N_n,\ \hat{\ell}(\cc_j,\bx_k) \geq \hat{\ell}(\cc_i,\bx_k). \label{eq:closer} \end{equation} Then 
	\begin{dmath}
\K(\s, \ell \  || \  \cc_j, \hat{\ell}  ; \bx_{1:n}) < 
\K(\s, \ell \  || \  \cc_i, \hat{\ell}  ; \bx_{1:n}) \\ \implies \K(\s, \hat{\ell} \  || \  \cc_j, \hat{\ell}  ; \bx_{1:n}) < 
\K(\s, \hat{\ell} \  || \  \cc_i, \hat{\ell}  ; \bx_{1:n}). \label{eq:overestimate}	
\end{dmath} \label{lem:comparison}
\end{lemma}
\begin{proof}
	For ease of notation, let $\hat{\omega}^i_k := \hat{\ell}(\cc_i,\bx_k)$, $z_k := \ell(\s,\bx_k)$ and $\hat{z}_k := \hat{\ell}(\s,\bx_k)$.
	Applying \eqref{eq:KLdiff}, we obtain
	\begin{align} &{\left[	\K(\s, \hat{\ell} \  || \  \cc_j, \hat{\ell}  ; \bx_{1:n})  - \K(\s, \hat{\ell} \  || \  \cc_i, \hat{\ell}  ; \bx_{1:n}) \right]} \nonumber  \\
	& - {\left[	\K(\s, \ell \  || \  \cc_j, \hat{\ell}  ; \bx_{1:n})  - \K(\s, \ell \  || \  \cc_i, \hat{\ell}  ; \bx_{1:n})\right]} \nonumber \\
	& =  \sum_{k=1}^n \mu \big( \hat{\omega}^i_k,\hat{\omega}^j_k,\hat{z}_k\big) - \sum_{k=1}^n \mu \big( \hat{\omega}^i_k,\hat{\omega}^j_k,z_k\big) \nonumber \\
	& = \sum_{k=1}^n (\hat{z}_k - z_k) \ln \left[ \frac{\hat{\omega}^i_k(1 - \hat{\omega}^j_k)}{\hat{\omega}^j_k(1 - \hat{\omega}^i_k)}\right]. \label{eq:Kdif}
	\end{align}
	By property \eqref{eq:envelope}, $\hat{z}_k \geq z_k$ for all $k$. Furthermore, the assumption \eqref{eq:closer} is equivalent to $\hat{\omega}^i_k \leq \hat{\omega}^j_k$ for all $k$, which then implies the RHS of \eqref{eq:Kdif} is non-positive. The result then follows. 
\end{proof}
\begin{proposition} \label{prop:wronglambda}
Let $(\bx_k)_{k \in \N }$ be an $n$-periodic sequence such that,
\begin{dmath} {\forall i \in \N_M,\ \exists j \in \B(\ell\mid\hat{\ell}),\ \forall k \in \N_n,} \\ {\hat{\ell}(\cc_j,\bx_k) \geq \hat{\ell}(\cc_i,\bx_k).} \label{eq:closeness} \end{dmath} Then $ \B(\hat{\ell}\mid\hat{\ell}) \subset \B(\ell\mid\hat{\ell})$.
\end{proposition}
\begin{proof}
	Suppose $i \in \N_M \setminus \B(\ell\mid \hat{\ell})$. Then by assumption \eqref{eq:closeness}, 
	$$ \exists j \in \B(\ell \mid \hat{\ell}),\ \forall k \in \N_n,\ \hat{\ell}(\cc_j,\bx_k) \geq \hat{\ell}(\cc_i,\bx_k). $$
	Definition \eqref{eq:Bdef2} implies,
	$$ \K(\s, \ell \  || \  \cc_j, \hat{\ell}  ; \bx_{1:n}) < \K(\s, \ell \  || \  \cc_i, \hat{\ell}  ; \bx_{1:n}).$$
	We can then apply Lemma \ref{lem:comparison} to obtain 
	$$ \K(\s, \hat{\ell} \  || \  \cc_j, \hat{\ell}  ; \bx_{1:n}) < \K(\s, \hat{\ell} \  || \  \cc_i, \hat{\ell}  ; \bx_{1:n}),$$
	which implies $i \notin \B(\hat{\ell}\mid \hat{\ell})$. We have shown ${ \B(\ell\mid\hat{\ell})^\complement \subset \B(\hat{\ell}\mid\hat{\ell})^\complement}$, which is equivalent to the result. 
\end{proof}
\begin{remark} \label{rem:conservative}
Proposition \ref{prop:wronglambda} can be interpreted as follows. Assuming the constraints \eqref{eq:closeness} on $\bx_k$ are satisfied, as long as $\hat{\ell}$ remains an envelope for $\ell$, the limiting support of the posterior cannot shrink compared to the case where $\ell = \hat{\ell}$. Thus, the algorithm behaves conservatively in the limit. The result holds even if $\ell$ is time-varying.
\end{remark}
\begin{remark} \label{rem:sensible}
When the envelope $\hat{\ell}$ is strictly decreasing with distance from the source, the assumption \eqref{eq:closeness} requires all the agents to be closer to some centre in $\B(\ell\mid\hat{\ell})$ than to the other centres. This suggests a sensible strategy would be to drive the agents towards the current MAP estimate (i.e. a maximiser of the posterior).
\end{remark}

\section{Numerical example and simulation results}
\label{sec:sim}
In this section, we examine a concrete example involving the 2D localisation of an electromagnetic source. 
We present simulation results to supplement the analytical results of the previous sections, and to demonstrate the effectiveness of the Bayesian estimation algorithm when used in the loop with a control law that guides the agents towards D-optimal measurement locations.
Consider a source antenna located at ground level, transmitting an RF signal of wavelength $\lambda$, with input power $P_T$ and effective area $A_T$. 
Suppose each agent is a UAV equipped with a receiving antenna of effective area $A_R$. 
Given source location $[\begin{smallmatrix} \s \\ 0 \end{smallmatrix}] \in \R^3$, the power received by an agent at location $[\begin{smallmatrix}
\x \\ z
\end{smallmatrix}] \in \R^3$ can be modelled by the Friis transmission formula \cite{friis1946}:
$$P_R(\s,\x) = \dfrac{A_R A_T P_T}{ \lambda^2( \| \s - \x \|^2 + z^2)}.$$
We simulate a team of four agents, constrained to fly at a constant altitude $z>0$, with positions in the plane that evolve according to
\begin{equation} \dot{ \x}_i(t) = \uu_i(t), \label{eq:dynamics} \end{equation}
where $\uu_i:[0,\infty) \to \R^2$ is the control signal applied to agent $i$. An agent at location $\bx_k \in \R^2$ reports a binary measurement $d_k$ by comparing the received power measured at time $t_k$ with a threshold $\eta>0$ according to
$$ d_k = \begin{cases}
0,& P_R(\s,\bx_k)  + W_k< \eta \\
1,& P_R(\s,\bx_k)  + W_k \geq \eta
\end{cases},$$
where $W_k \sim \mathcal{N}(0,\sigma^2)$ accounts for sensor noise. 
The probability-of-detection function is therefore given by
$$ \ell(\s,\x) = \mathcal{Q}\left(\dfrac{\eta - P_R(\s,\x)}{\sigma} \right),$$
where 
$$ \mathcal{Q}(x) : = \frac{1}{\sqrt{2\pi}}\int_x^\infty e^{-\frac{u^2}{2}} du  $$
is the Q-function. 

Suppose there is a maximum transmission delay of $\frac{\tau}{2} \geq 0$ seconds between the agents and the fusion centre.
We ignore the effects of packet drop, and constrain the agents to take measurements synchronously every $T>\tau$ seconds. Recalling that subscripts are assigned according to the order in which the measurements arrive at the fusion centre, this implies
$$ 0 \leq t_1 = ... = t_4 < t_5 = t_6 = ...,$$
where $t_{k+4} - t_k = T$. The fusion centre processes all four measurements pairs $(\bx_k,d_k),...,(\bx_{k+3},d_{k+3})$ within the interval $[t_k,t_k + \frac{\tau}{2}]$. It then computes the posterior mean
$$ \bar{\s}_k :=  \sum_{j=1}^M \hat{p}_k(i \mid d_{1:k} ; \bx_{1:k})\cc_j,$$
which it transmits back to the agents. All agents then receive $\bar{\s}_k$ no later than $t_k + \tau$, at which time they synchronously update their local copies of the mean $\hat{\s}_i(t_k + \tau) = \bar{\s}_k$. Thus, these local copies evolve in continuous-time according to
$$ \hat{\s}_i(t) := \bar{\s}_{\kappa(t - \tau)},$$ where $\kappa:[0,\infty) \to \N,\ \kappa(t) := \max \{ k \mid t_k \leq t \}.$
The posterior mean is used as an input to the controller proposed below:
\begin{equation} \uu_i(\x_i,\hat{\s}_i) = - \left( \x_i - \hat{\s}_i - \bd_i \right) \label{eq:econtrol} \end{equation}
where $\bd_i:= r\begin{bmatrix} \cos \theta_i \\ \sin \theta_i \end{bmatrix}$ is chosen according to \eqref{eq:optimalangles} - \eqref{eq:optimalr}. This drives the agents towards the optimal locations dictated by \eqref{eq:exactsoln}. 

For the simulations below, the search region $S$ is a 75 m $\times$ 75 m planar region (at ground level), and the source location is sampled from the uniform distribution over $S$. We choose $M$ centres, aligned in a uniform grid over a 100 m $\times$ 100 m region containing $S$ at the centre. Parameter values are $A_R = A_T = 1\ \mathrm{m}^2$, $P_T = 1$ W, $\lambda = 1$ m, $z = 10$ m, $\eta = 5 \times 10^{-3}$ W, $\sigma = 2.5 \times 10^{-3}$ W, $T = 0.04$ s and $\tau = 0.02$ s, unless otherwise stated. Each agent was initialized as shown in Figure \ref{fig:trajectories}. A uniform prior for the source location was used to initialize the Bayesian updates.

To numerically examine the effects of discretisation on estimation performance, for every $M \in \{10^2, 20^2,...,50^2\}$ we run 100 Monte Carlo trials and compute the RMS estimation error $e_k$ by averaging $\|\bar{\s}_k - \s \|$. A supplementary animation of the simulation for $M = 30^2$ is available at \texttt{\url{https://youtu.be/l8Awf0KCt4s}}. The results are plotted in Figure \ref{fig:still} for up to $k = 1000$ measurements. 
In practice, the entropy $h_k$ of the posterior is a good indicator of convergence. We approximately evaluate $e_\infty$ by computing the RMS error at $k = 1000$, averaging only the trials for which $h_{1000} < 1$ nat. These results are recorded in Table \ref{tab: stats}, and we observe that $e_\infty$ decreases monotonically with grid spacing.
\begin{table}[h]
	\centering
	\begin{threeparttable}
		\caption{Asymptotic estimation error.}
		\label{tab: stats}
		\begin{tabular}{|c|c|c|c|c|c|}
			\hline
			$M=$ & $ 10^2$ & $ 20^2$ & $ 30^2$ & $ 40^2$ & $ 50^2$ \\ \hline 
			Grid Spacing (m) & 10 & 5 & 3.33 & 2.5 & 2 \\ 
			Approx. $e_\infty$ (m) & 3.95 & 1.96 & 1.43 & 1.04 & 0.84 \\
			Trials with $h_{1000}<1$ & 99 \% & 97 \% & 96 \% & 97 \% & 98 \% \\ 
			\hline
		\end{tabular}
	\end{threeparttable}
\end{table}
\begin{figure*}[p]
	\centering
	\begin{subfigure}{1.7\columnwidth}
		\centering
		\includegraphics[width=\linewidth]{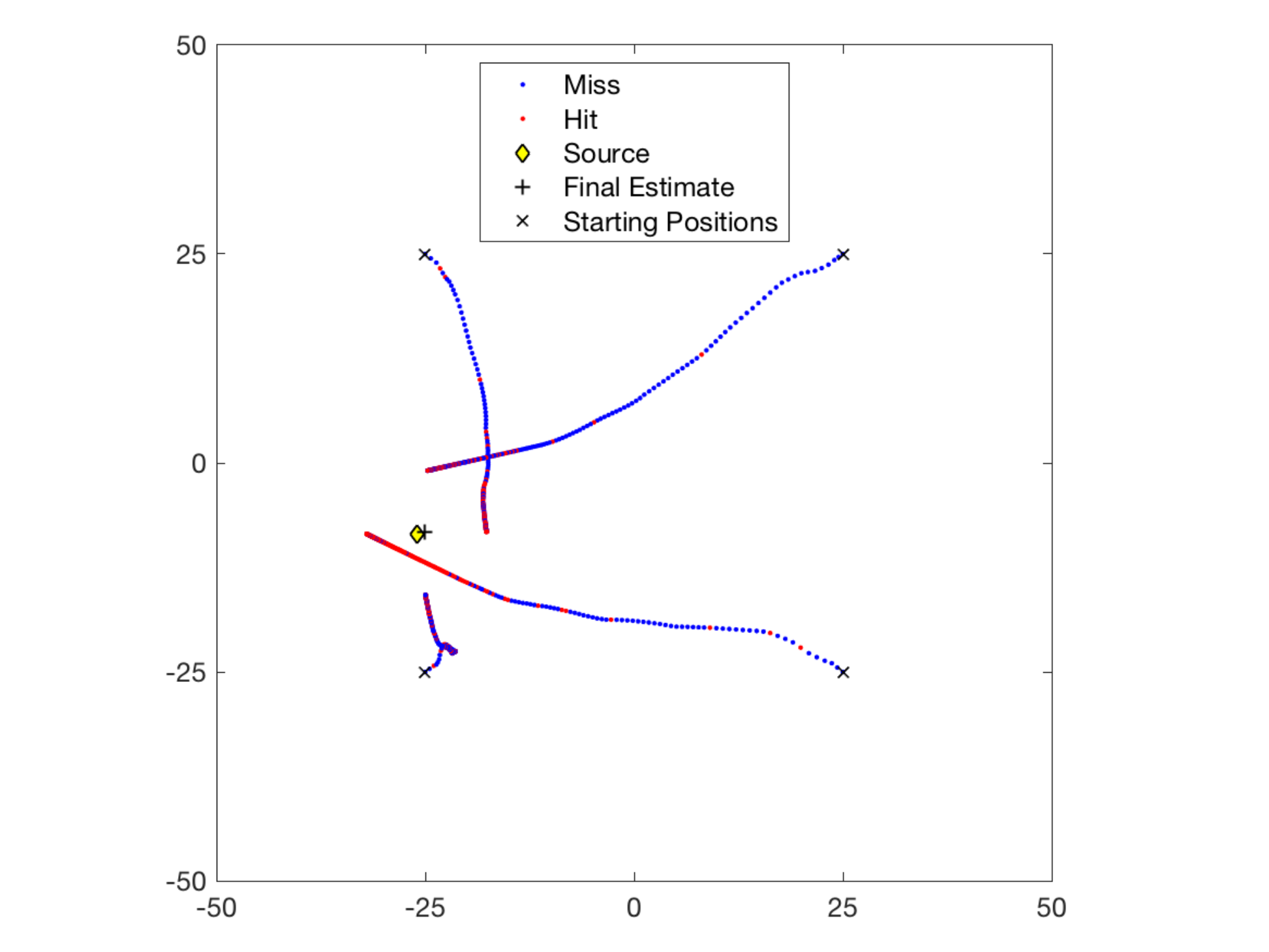}
		\caption{History of measurement pairs ($M = 30^2$, units: m).}
		\label{fig:trajectories}
	\end{subfigure}
	\begin{subfigure}{1.7\columnwidth}
	\centering
	\includegraphics[width=\linewidth]{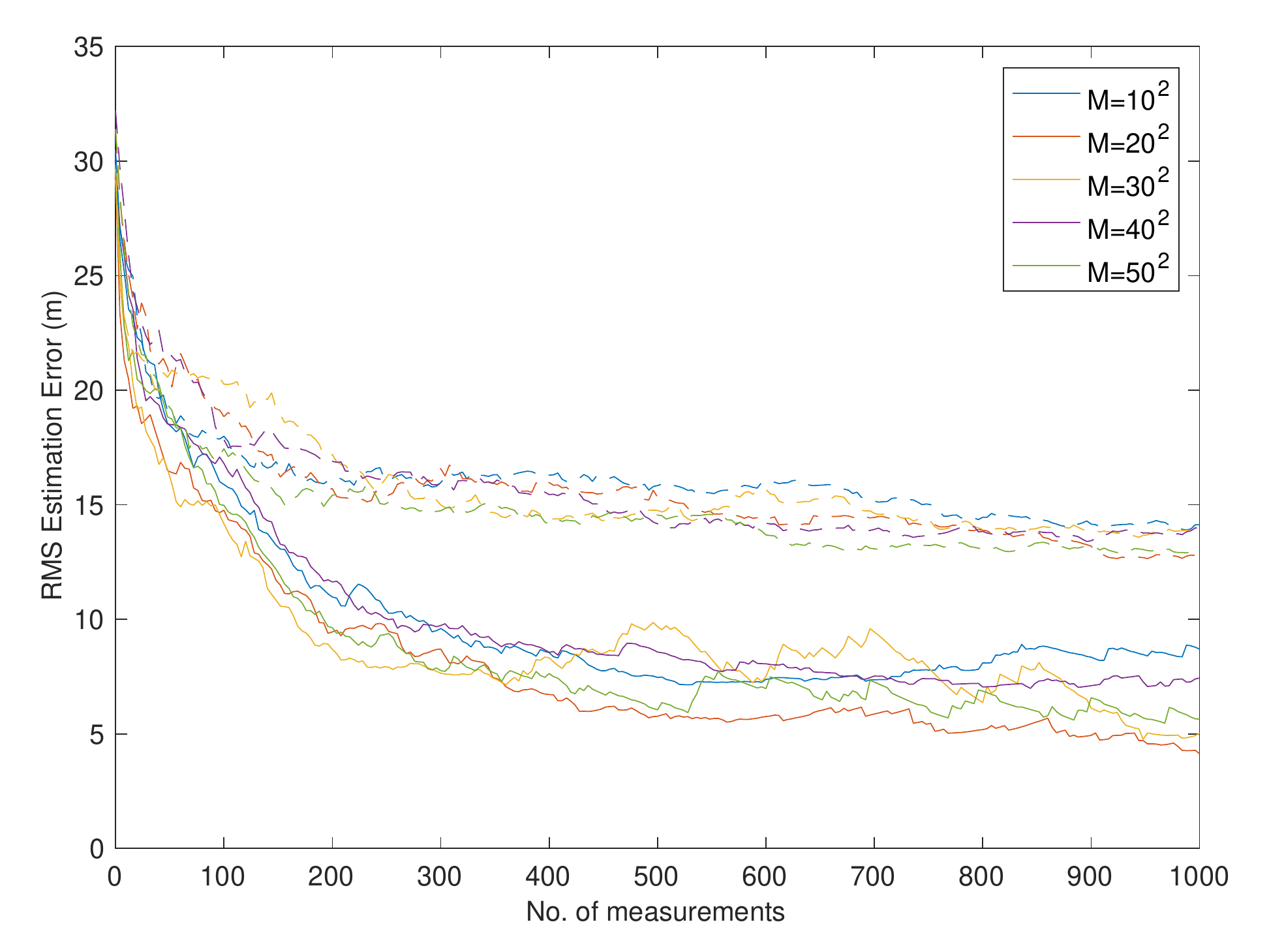}
	\caption{With controller (solid lines) vs. without controller (dashed lines)}
	\label{fig:still}
\end{subfigure}
\caption{Simulation results}
\end{figure*}
\begin{figure*}[p]
	\begin{subfigure}{1.7\columnwidth}
		\centering
		\includegraphics[width=\linewidth]{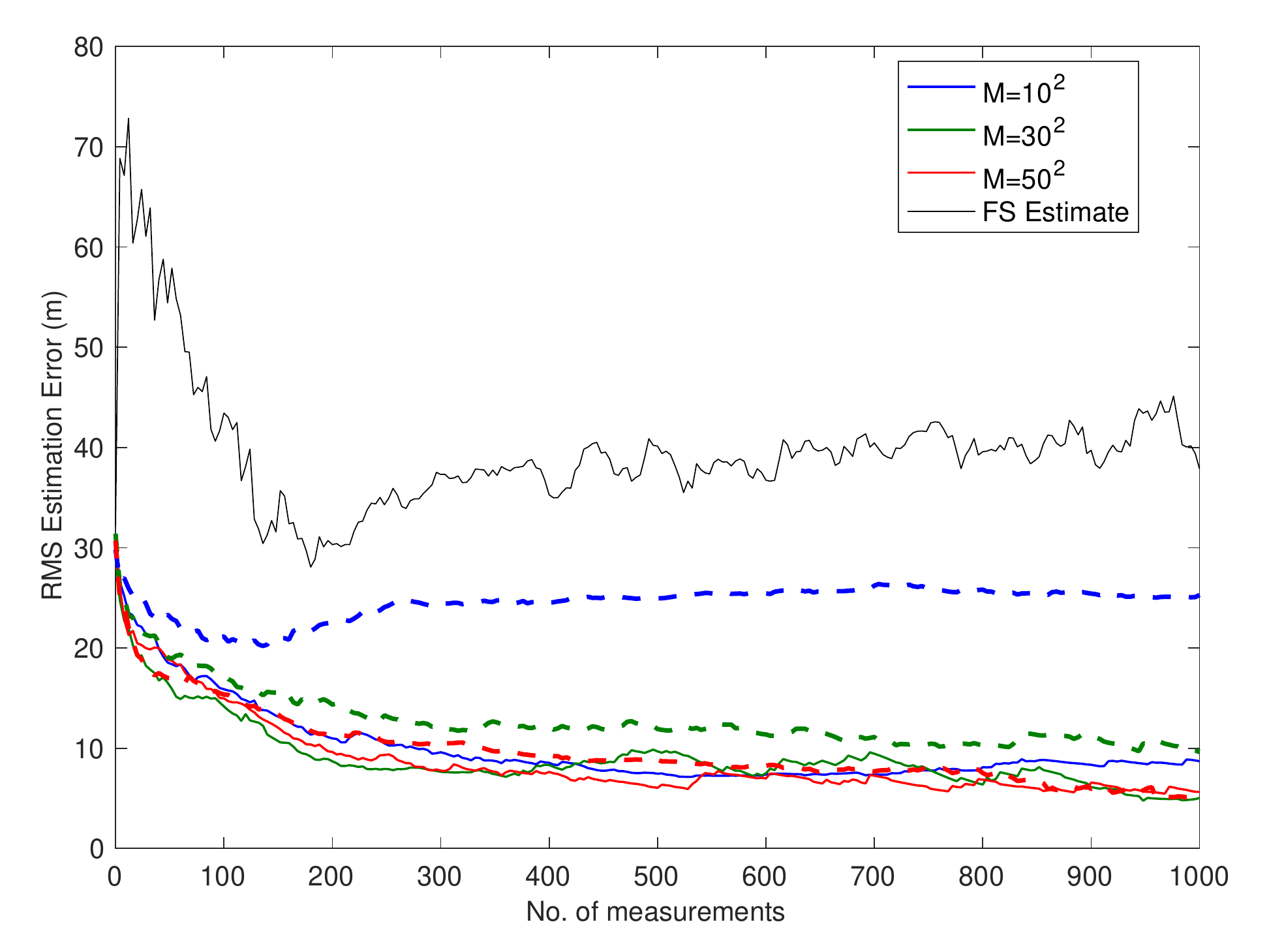}
		\caption{Discretised Bayesian estimator (solid coloured lines) vs. FS vs. SIR with $M$ particles (dashed coloured lines)}
		\label{fig:bayes}
	\end{subfigure} 
\begin{subfigure}{1.7\columnwidth}
	\centering
	\includegraphics[width=\linewidth]{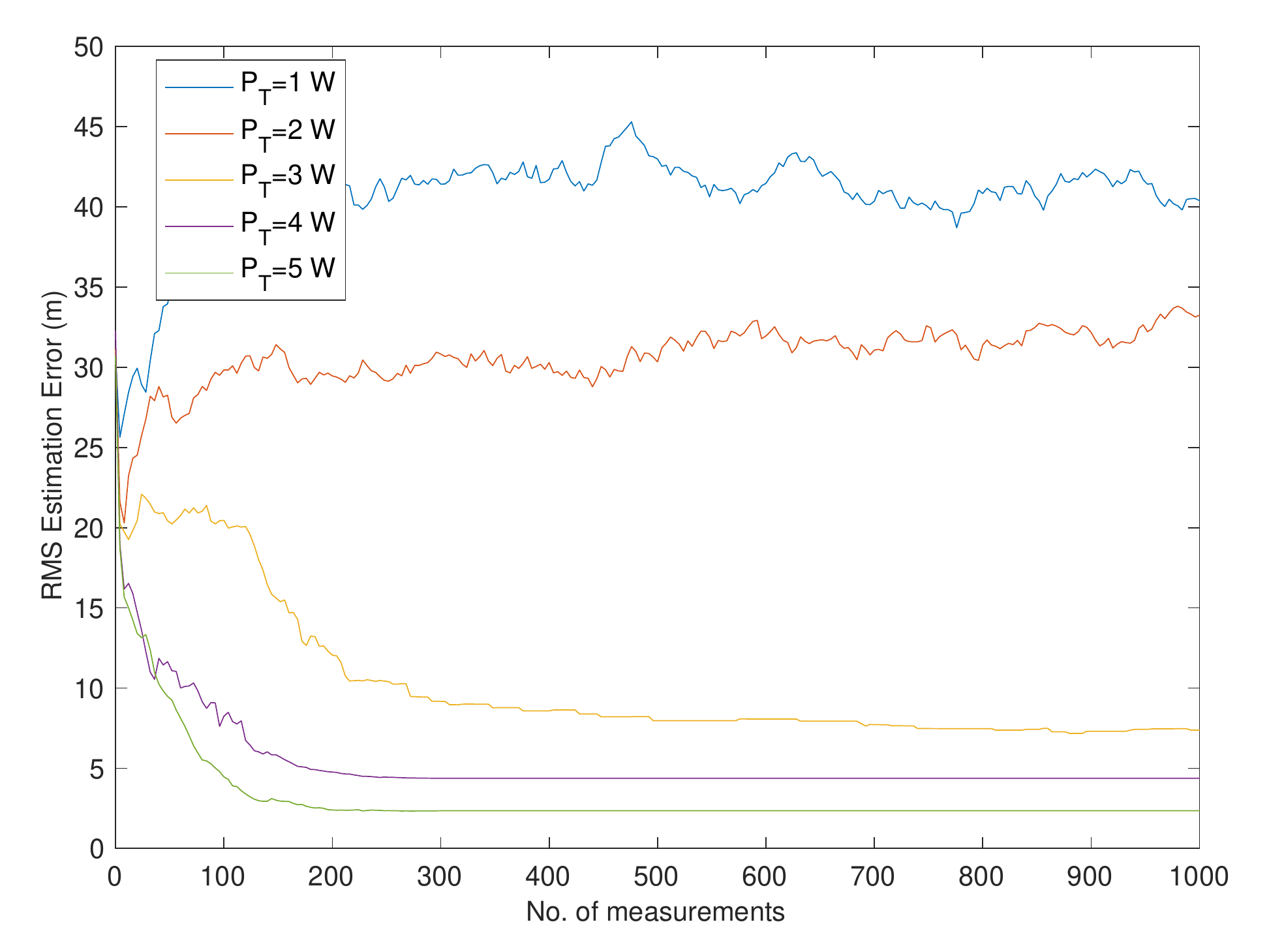}
	\caption{Effect of changing $P_T$. (\hbox{$\hat{P}_T =5$ W}, $M = 20^2$)}
	\label{fig:imperfect}
\end{subfigure}
\caption{Simulation results}
\end{figure*}
\subsection{Comparison with other approaches}
As mentioned in the introduction, we employ the same Bayesian estimation algorithm adopted in \cite{vergassola2007, ristic2016}. Here, we compare its performance with the Fisher scoring (FS) approach of \cite{vijayakumaran2007} and the particle filtering approach in \cite{ozdemir2008}. 

The complexity of the full ML estimator in \cite{vijayakumaran2007} is $O(q^3k^2)$ for the $k$th measurement. The required memory grows linearly with $k$. Since this can be impractical for real-time processing, \cite{vijayakumaran2007} also proposes a real-time approximation, which has complexity $O(q^3)$ per measurement and requires constant memory.
In contrast~\cite{selvaratnam2017b} establishes the complexity and memory requirements of our Bayesian method as both of $O(M)$, constant with respect to the number of measurements. Recall that $M$ is the number of centres.   
The RMS error obtained by implementing the full ML estimator in \cite{vijayakumaran2007} is plotted in Figure \ref{fig:bayes}.
This ML estimator is based on the Newton method, and since the log-likelihood function for this problem is non-concave, there are no convergence guarantees as $k \to \infty$. In the simulations, five Newton iterations are performed per measurement. 
To facilitate a meaningful comparison between the two approaches, the control-law \eqref{eq:econtrol} is also employed to direct the agents under the ML approach, but with the posterior mean $\hat{\s}$ replaced by the ML estimate.

Ozdemir et al.~\cite{ozdemir2008} employ a particle filter to estimate the location of a moving source using binary measurements. Specifically, they adopt sequential importance re-sampling (SIR)~\cite[Algorithm 4]{arulampalam2002a}. 
As noted in the introduction, and explained in \cite[Footnote 5]{arulampalam2002a}, such an approach is not well suited to a stationary source. In particular, the re-sampling step introduces additional complexity without improving performance. The numerical results plotted in Figure \ref{fig:bayes} support this claim. 

 \subsection{Effect of control strategy}
 We now examine the effectiveness of the control strategy \eqref{eq:econtrol}, which drives the agents into the D-optimal geometries defined by \eqref{eq:exactsoln}. In Figure \ref{fig:still}, the RMS estimation error of the Bayesian algorithm with the control-law in the loop is plotted against a scenario in which the agents remain fixed at their initial positions (which are evenly spaced throughout the environment). The motion of the agents significantly increases the rate at which the estimation error decays. A sample trajectory induced by the control law is plotted in Figure \ref{fig:trajectories}, along with the full history of measurement pairs $(\bx_k,d_k)$. 
 
 \subsection{Effect of inexact knowledge of probability of detection}
 To examine the effect of inexact knowledge of $\ell$, we fix the assumed value of the transmitted power at $\hat{P}_T = 5 W$, and vary the true value $P_T$ between 1  W and 5 W. Thus the assumed probability-of-detection function $\hat{\ell}$ remains an envelope for $\ell$. 
 We use $M = 20^2$ grid points for the simulations. Consistent with the strategy proposed in Remark \ref{rem:sensible}, we modify the control law \eqref{eq:econtrol} by replacing $\hat{\s}$ with the MAP estimate.
 This time, the angles $\theta_i$ are chosen according to \eqref{eq:optimalangles}, but we set $r = 2.5$ m so that the agents are driven to converge to points that are closer to the MAP estimate than the other centres. The results of 100 Monte Carlo trials are plotted in Figure \ref{fig:imperfect}. We observe a graceful degradation in RMS estimation error for $P_T \geq 3$ W, but the estimator ceases to be effective when the transmitted power falls to 2 W or less.   

\section{Conclusion}
The localisation of a stationary source using binary measurements is addressed in this paper. The adopted estimation procedure discretises the search region into a finite set of centres, and uses a Bayesian update rule to maintain a posterior over these centres. A theoretical analysis of this discrete posterior is presented. Conditions on the sequence of measurement locations are derived which guarantee posterior consistency when the source is coincident with a centre. The more general case of an arbitrarily located source is studied by restricting attention to periodic measurement location sequences. In this case, the algorithm asymptotically selects the indices of centres which minimise KL divergence from the true measurement probability distribution. The results described above hold for general, continuous probability-of-detection functions. Specific results are also derived for range-dependent probability-of-detection functions. 
 
 The design of D-optimal measurement locations with respect to the Bayesian Information Matrix is also formulated mathematically. Although obtaining an analytic solution is intractable, a relaxed version of the problem is proposed, which maximises the Fisher Information determinant about the expected source location. 
 The FIM for a range-dependent probability of detection is then derived, and a closed-form solution is established to a constrained version of the resulting optimisation problem.
The effect of having inexact knowledge of the probability-of-detection function is examined by assuming knowledge of an envelope for the function. Under certain conditions, the asymptotic support of the posterior is shown to be no smaller than when the true probability-of-detection function coincides with the assumed envelope. Finally, a numerical example is simulated to supplement the theoretical results. A control strategy is proposed to guide the agents into the D-optimal measurement locations, and a comparison of the algorithm with the approaches of \cite{vijayakumaran2007,ozdemir2008} is also presented.
 
There are several promising directions for future work, including extending the algorithm to deal with multiple and/or moving targets. The closed-loop properties of the system should also be studied theoretically, and control strategies developed that are time-optimal and/or guarantee consistency. The works \cite{liu2013a, masazade2012} may also offer insight into designing iterative methods for solving \eqref{eq:fullproblem} directly. 
Finally, a distributed implementation of the algorithm should be considered, while incorporating the effects of transmission delay, asynchronous updates, and packet drop.

\label{sec:conc}
\section*{Acknowledgement}
The authors would like to thank Brian D.O. Anderson for his insightful comments, which inspired some key results in this paper.
{\small 
\bibliographystyle{ieeetr}
\bibliography{Bibliography/searchx}}

\appendices
\section{Infinite products of random variables}
\label{app:infiniteprod}
This appendix contains some results concerning the convergence of infinite products of random variables, on which the rest of the paper relies. \\

\begin{lemma}
	\label{lem:zeromean}
	Let $(W_k)_{k \in \N}$ be a sequence of independent random variables such that $\Ex[W_k] = 0$ for all $k$. If the sequence is bounded, that is $$  \exists M > 0 \text{ s.t. } \forall k,\ |W_k|<M, $$ then
	$$ \forall p > \frac{1}{2}, \ \lim_{n \to \infty} \dfrac{1}{n^p} \sum_{k=1}^n W_k  = 0 \text{ a.s.. }$$
\end{lemma}
\begin{proof}
	The sequence $W_k$ is bounded, and therefore there exists $C>0$ such that $\var[ X_k] \leq C$ for all $k$. Let $p > \frac{1}{2}$, and note that $\Ex \left[ k^{-p} W_k \right]=0$ and $\var \left[ k^{-p} W_k \right] \leq k^{-2p}C$. This implies 
	$$\forall n,\ \sum_{k=1}^n \var \left[ \frac{W_k}{k^p} \right] \leq \sum_{k=1}^n \frac{C}{k^{2p}}. $$
	Now $2p>1$, and therefore $ \sum_{k=1}^\infty \var \left[ k^{-p}W_k \right]  < \infty $ by \cite[Theorem 3.28]{rudin1964-1}.
	Applying \cite[Theorem 12.2]{williams1991} yields $$\sum_{k=1}^\infty \frac{W_k}{k^p} < \infty \text{ a.s.},$$ and Kroneckers' Lemma \cite[Lemma 12.7]{williams1991} then implies the result.
\end{proof}

\begin{corollary}
	\label{cor:limsups}
	Let $(X_k)_{k \in \N}$ be a sequence of independent, bounded random variables. Then 
	$$ \forall p > \frac{1}{2}, \ \limsup_{n \to \infty}  \dfrac{1}{n^p} \sum_{k=1}^n X_k = \limsup_{n \to \infty}  \dfrac{1}{n^p} \sum_{k=1}^n \Ex[X_k] \text{ a.s..}$$
\end{corollary}
\begin{proof}
	Letting $W_k := X_k - E[X_k]$, the result follows from Lemma \ref{lem:zeromean}. 
\end{proof}

\begin{lemma}
	\label{lem:infproducts}
	Let $(Z_k)_{k\in\N}$ be a sequence of independent random variables for which there exist $\alpha,\beta > 0$ such that $Z_k \in [\alpha,\beta]$ for all $k$. If there exists $p > \frac{1}{2}$ such that \begin{equation} \limsup_{n \to \infty} \dfrac{1}{n^p} \sum_{k=1}^n \Ex[\ln Z_k] < 0, \label{eq:ElnZk} \end{equation} then $$ \lim_{n \to \infty} \prod_{k=1}^n Z_k = 0 \text{ a.s. } $$
\end{lemma}
\begin{proof}
	The random variable $\ln Z_k \in [ \ln \alpha, \ln \beta] \subset (0,\infty)$. For any $p > \frac{1}{2}$, 
	\begin{align*} \limsup_{n \to \infty} \frac{1}{n^p} \ln \left( \prod_{k=1}^n Z_k \right) & = \limsup_{n \to \infty} \frac{1}{n^p} \sum_{k=1}^n \ln Z_k \\
	& = \limsup_{n \to \infty} \frac{1}{n^p} \sum_{k=1}^n \Ex[ \ln Z_k ] \text{ a.s. }\end{align*}
	by Corollary \ref{cor:limsups}. Suppose \eqref{eq:ElnZk} holds, and fix a realisation $(Z_k)_{k \in \N} $ for which $$c := \limsup_{n \to \infty} \frac{1}{n^p} \sum_{k=1}^n \Ex[ \ln Z_k ]< 0.$$
	This implies there exists $N \in \N$ such that $\frac{1}{n^p} \ln \left( \prod_{k=1}^n Z_k \right)< \frac{c}{2} < 0$ for all $n > N$, and thus 
	$$\forall n > N,\  \ln \left( \prod_{k=1}^n Z_k \right) < \dfrac{n^p c}{2}<0.$$ This in turn implies $\ln \left( \prod_{k=1}^n Z_k \right) \to - \infty$ as $n \to \infty$, which yields the result. 
\end{proof}

\begin{lemma}
	\label{lem:hoeffding}
	Let $(Z_k)_{k\in\N}$ be a sequence of independent random variables for which there exist $\alpha,\beta > 0$ such that $Z_k \in [\alpha,\beta]$ for all $k$.  Let $\gamma_n :=  \sum_{k=1}^n \Ex \left[ \ln Z_k\right]$. If 
	\begin{equation} \frac{\gamma_n}{\sqrt n} \to - \infty, \label{eq:gammadown} \end{equation}
	then for any $\epsilon > 0$, there exists $K \in \N$ such that
	\begin{equation}\forall n > K,\ \mathrm{ Pr} \left( \prod_{k=1}^n Z_k \geq \epsilon \right) \leq \exp \left( -\dfrac{2 \left( \ln\epsilon - \gamma_n \right)^2}{n( \ln\beta - \ln \alpha)^2} \right). \label{eq:rate} \end{equation}
\end{lemma}
\begin{proof}
	Let $X_k := \ln( Z_k)$, and define $S_n := \sum_{k=1}^n X_k$. The $X_k$ are independent, and therefore $$\Ex[ S_n] = \sum_{k=1}^n \Ex \left[ \ln Z_k\right]= \gamma_n.$$
	Note that $\prod_{k=1}^n Z_k = \exp(S_n)$, and therefore
	$$ \left( \prod_{k=1}^n Z_k \geq \epsilon \right) \iff \left( S_n \geq \ln \epsilon \right), $$ where $\epsilon > 0$. Therefore for any $\epsilon > 0$,
	\begin{equation} \mathrm{ Pr} \left(  \prod_{k=1}^n Z_k \geq \epsilon \right) = \mathrm{ Pr} \left( \frac{S_n - \gamma_n}{n} \geq \frac{\ln \epsilon  - \gamma_n}{n} \right). \label{eq:subsetineq}  \end{equation}
	If \eqref{eq:gammadown} holds, then
	$$ \forall \epsilon > 0, \ \exists K \in \mathbb{N} \text{ s.t. } \forall n > K, \ \gamma_n < \ln \epsilon ,$$ and therefore $\dfrac{ \ln \epsilon - \gamma_n}{n} > 0$ for all $n > K$. Furthermore, \hbox{$X_k \in [ \ln(\alpha), \ln( \beta)]$} for all $k$. We can therefore apply Hoeffding's inequality \cite[Theorem 2]{hoeffding1963} for all $n \geq K$: 
	$$ \mathrm{ Pr} \left( \frac{S_n - \gamma_n}{n} \geq \frac{\ln(\epsilon) - \gamma_n}{n} \right) \leq \exp \left( -\dfrac{2( \ln\epsilon - \gamma_n)^2}{n( \ln\beta - \ln \alpha)^2} \right).$$
	Combining this with \eqref{eq:subsetineq} gives us \eqref{eq:rate}. 
\end{proof}
\section{Fisher Information}
\label{app:Fish}
A general expression for the FIM in Section \ref{sec:fish} is derived below, based on the likelihood function $g$ defined in \eqref{eq:gfun}.
The log-likelihood gradient is given by
\begin{align*}
\nabla_\s \ln g(d_k \mid \s ; \bx_k) & = \frac{ \nabla_\s g(d_k \mid \s ; \bx_k)}{g(d_k \mid \s ; \bx_k)} \\
& = \left[ \frac{ \nabla_\s \ell(\s,\bx_k)}{\ell(\s,\bx_k)}\right]^{d_k} \left[\frac{ \nabla_\s \ell(\s,\bx_k)}{\ell(\s,\bx_k) - 1} \right]^{1 - d_k} \\
& = \frac{ \nabla_\s \ell(\s,\bx_k)}{(-1)^{1 - d_k}[\ell(\s,\bx_k)]^{d_k}[1-\ell(\s,\bx_k)]^{1 - d_k}},
\end{align*}
and its Hessian,
\begin{dmath*}
	{ \nabla^2_\s \ln g(d_k \mid \s ; \bx_k)}  = {\frac{ \partial}{\partial \s}\left[ \nabla_\s \ln g(d_k \mid \s ; \bx_k) \right]} \\
	=  {\left[ \frac{ \ell(\s,\bx_k) \nabla^2_\s \ell(\s,\bx_k) - \nabla_\s \ell(\s,\bx_k)\nabla^\top_\s \ell(\s,\bx_k)  }{\ell(\s,\bx_k)^2}\right]^{d_k}} \cdot {\left[ \frac{ [\ell(\s,\bx_k) - 1] \nabla^2_\s \ell(\s,\bx_k) - \nabla_\s \ell(\s,\bx_k)\nabla^\top_\s \ell(\s,\bx_k)  }{[ 1 - \ell(\s,\bx_k)]^2}\right]^{1-d_k} } \\
	= {\dfrac{\nabla^2_\s \ell(\s,\bx_k)}{(-1)^{1 - d_k}[\ell(\s,\bx_k)]^{d_k}[1-\ell(\s,\bx_k)]^{1 - d_k}}}\\ {-
		\dfrac{\nabla_\s \ell(\s,\bx_k)\nabla_\s \ell(\s,\bx_k)^\top}{[\ell(\s,\bx_k)]^{2d_k}[1-\ell(\s,\bx_k)]^{2(1 - d_k)}}}.
\end{dmath*}
The FIM for a single reading is then
\begin{align*}
& \I(\s;\bx_k)  = - \Ex \left[ \nabla^2_\s \ln g(d_k \mid \s ; \bx_k) \right] \\
 & = -\sum_{d\in \{0,1\}}{g(d \mid \s;\bx_k)} \nabla_\s^2 \ln {g(d \mid \s;\bx_k)} \\
& = -  \ell(\s,\bx_k)  \left[
\frac{ \ell(\s,\bx_k) \nabla^2_\s \ell(\s,\bx_k) - \nabla_\s \ell(\s,\bx_k)\nabla^\top_\s \ell(\s,\bx_k)  }{\ell(\s,\bx_k)^2} \right] \\
& - [1 - \ell(\s,\bx_k)] \left[ \frac{ [\ell(\s,\bx_k) - 1] \nabla^2_\s \ell(\s,\bx_k) - \nabla_\s \ell(\s,\bx_k)\nabla^\top_\s \ell(\s,\bx_k)  }{[ 1 - \ell(\s,\bx_k)]^2} \right] \\
& = \dfrac{\nabla_\s \ell(\s,\bx_k)\nabla_\s \ell(\s,\bx_k)^\top}{\ell(\s,\bx_k)[1 - \ell(\s,\bx_k)]}.
\end{align*}
\end{document}